\documentclass[11pt]{article}

\usepackage[margin=1in]{geometry}
\usepackage[utf8]{inputenc}
\usepackage[T1]{fontenc}
\usepackage{lmodern}
\usepackage{microtype}
\usepackage{amsmath,amssymb,amsthm,mathtools,bm}
\usepackage{booktabs,tabularx,array}
\usepackage{graphicx}
\usepackage{xcolor}
\usepackage[colorlinks=true,linkcolor=blue!50!black,urlcolor=blue!50!black,
            citecolor=blue!50!black]{hyperref}

\allowdisplaybreaks

\theoremstyle{plain}
\newtheorem{theorem}{Theorem}[section]

\newtheorem{proposition}[theorem]{Proposition}
\newtheorem{corollary}[theorem]{Corollary}
\theoremstyle{definition}

\theoremstyle{remark}

\DeclareMathOperator{\E}{\mathbb E}
\DeclareMathOperator{\Var}{Var}
\DeclareMathOperator{\Tr}{Tr}
\DeclareMathOperator{\Ran}{Ran}

\DeclareMathOperator{\diag}{diag}
\newcommand{\id}{\mathbf 1}
\newcommand{\R}{\mathbb R}
\newcommand{\Z}{\mathbb Z}
\newcommand{\Circ}{\mathbb S^1}
\newcommand{\hilb}{\mathcal H}
\newcommand{\cE}{\mathcal E}
\newcommand{\cF}{\mathcal F}
\newcommand{\cL}{\mathcal L}
\newcommand{\eps}{\varepsilon}          
\newcommand{\Estar}{E_{\star}}          
\newcommand{\h}{\hbar}
\newcommand{\half}{\tfrac12}
\newcommand{\wt}{\widetilde}
\newcommand{\dd}{\,\mathrm d}

\newcommand{\journalfigure}[2][]{%
  \IfFileExists{#2}{\includegraphics[#1]{#2}}{%
    \fbox{\begin{minipage}[c][2.1in][c]{0.90\linewidth}
      \centering Figure file not found:\\[2mm]
      \texttt{\detokenize{#2}}
    \end{minipage}}}}

\title{The Decoherence Exponent:\\
Stable Phase Noise and Constraints on Objective State Reduction}

\author{%
Milen V.\ Velev\\
Burgas State University ``Prof.\ Dr.\ Asen Zlatarov'', Burgas, Bulgaria\\
\texttt{milen.velev@gmail.com}
\and
Ivo D.\ Dinov\\
Statistics Online Computational Resource and\\
Michigan Center for Applied and Interdisciplinary Mathematics\\
University of Michigan, Ann Arbor, MI, USA\\
\texttt{statistics@umich.edu}}
\date{\today}

\begin{document}
\maketitle

\begin{abstract}
Let $L_u$ be a real-valued symmetric L\'evy process representing the lift of an
unobserved phase, and let a quantum system couple to it through a self-adjoint
charge $Q$. Averaging the random unitary $e^{iL_uQ}$ produces a completely
positive, trace-preserving dephasing semigroup. If the characteristic exponent
has no intrinsic charge scale in the precise sense
$\eta(\lambda\xi)=c(\lambda)\eta(\xi)$, then continuity together with the
L\'evy--Khintchine representation forces
$\eta(\xi)=D\lvert\xi\rvert^{\alpha}$ with $0<\alpha\le2$, so that coherences
between charge sectors decay at rate
$\Gamma_{ab}=D\lvert q_a-q_b\rvert^{\alpha}$. When all charge differences are
integers the construction descends to the wrapped process
$\Theta_u=L_u\bmod 2\pi$ on $\mathbb S^1$. For arbitrary real charges the lift,
not the wrapped variable, is the relevant random phase. Schoenberg's theorem
gives complete positivity for every finite real charge spectrum when
$0<\alpha\le2$, and for every $\alpha>2$ three equally spaced charges already
provide an analytic obstruction. The boundary is therefore exact, and the
finite-matrix diagonalizations reported here serve as numerical checks of it.

A Gaussian integrated phase always gives quadratic charge dependence, although
colored Gaussian noise need not produce a Markov semigroup. Two idealized
non-Gaussian mechanisms are analyzed: inverse-power Poisson shot noise, which
under stated convergence and symmetry assumptions gives $\alpha=d/p$, and
Bochner subordination of Brownian phase by an $\alpha/2$-stable clock. For a
separate bounded reduction walk we prove Born probabilities for every
nondegenerate symmetric bounded proposal law, and simulations with
symmetrically truncated stable parent laws of index $0.5,1,1.5,2$ confirm the
law-independence. In continuous time the two natural readings of an
$\alpha$-stable driver both fail to reduce. The exact-flow (Marcus) equation
oscillates for every $\alpha$, and the naive It\^o jump equation is not state
preserving. For Gaussian white noise the It\^o model collapses with Born
probabilities, whereas the exact-flow (Stratonovich/Wong--Zakai) model does not
collapse without an imposed threshold and yields non-Born threshold-exit
probabilities. The exponent is accessible through ratios of rates at different
splittings, independent of the overall scale. The quadratic member reproduces
the double-commutator small-step limit of Milburn's model, not Milburn's exact
Poissonian dynamics. A bi-temporal geometry is examined only as a possible
source of a compact phase. Closed timelike curves, nonunitary slice evolution,
an unstable nonzero-mode tower, and the absence of a derived charge assignment
prevent it from constituting a consistent field-theoretic realization.
\end{abstract}


\newpage

\section{Introduction}
\label{sec:intro}

Unitary Schr\"odinger evolution preserves superpositions, whereas measurements
return definite outcomes with Born-rule frequencies. Dynamical-reduction models
address this tension by adding stochastic and nonlinear state-vector dynamics.
The Ghirardi--Rimini--Weber model, continuous spontaneous localization, and
relativistic variants make that program quantitative
\cite{ghirardi1986,ghirardi1990,tumulka2006}; broad reviews include
\cite{bassi2013,carlesso2022}. Relativistic consistency and operational
no-signaling remain central constraints on any such modification
\cite{aharonov1980,aharonov1981,fleming1988,myrvold2002,gisin1990,polchinski1991}.

A logically distinct question concerns a system coupled to a phase-like degree
of freedom that is not observed. Once a probability law for that variable is
specified, averaging the corresponding random unitaries defines a reduced
channel. Irreversibility does not follow merely from calling the variable
``hidden''. Here it follows from the assumed convolution-semigroup law of its
increments. The resulting dephasing rate is controlled by the characteristic
exponent of the phase increment. Gaussian phase fluctuations give a rate
quadratic in the coupled charge difference. Colored Gaussian noise can alter
the time dependence and can make the dynamics non-Markovian, but it cannot
alter that quadratic charge dependence.

This paper investigates the result of dropping the Gaussian assumption. The primary object
is a real-valued L\'evy lift $L_u$ of a phase. The operational parameter $u$ is
laboratory evolution time throughout the phenomenology; it is not the radial
variable $t$ of the geometric discussion in
Section~\ref{sec:discussion_geometry}. If the characteristic exponent of $L_u$
is scale covariant, then it is necessarily symmetric stable,
\begin{equation}
  \E\!\left[e^{i\xi(L_u-L_0)}\right]
  =e^{-Du\lvert\xi\rvert^{\alpha}},
  \qquad 0<\alpha\le2,
  \label{eq:stable_cf_intro}
\end{equation}
and averaging $e^{iL_uQ}$ gives
\begin{equation}
  \Gamma_{ab}=D\lvert q_a-q_b\rvert^{\alpha}.
  \label{eq:master_rate}
\end{equation}
The compact circle enters by wrapping $L_u$ modulo $2\pi$. That wrapped process
fully describes the coupling only when conjugation by $e^{i\theta Q}$ is
single-valued on the circle, equivalently when all charge differences lie in
$\Z$. Arbitrary real energy charges require the lifted real phase. This
distinction is essential both mathematically and physically. The scale
covariance that selects Eq. \eqref{eq:stable_cf_intro} cannot even be stated on
the integer Fourier modes of a circle, since $\lambda k\notin\Z$ for generic
$\lambda>0$.

The Gaussian endpoint of Eq. \eqref{eq:master_rate} is the familiar
double-commutator energy-dephasing law. It agrees with the leading small-step
expansion of Milburn's exact intrinsic-decoherence equation, but not with that
exact equation at finite step size \cite{milburn1991,diosi2005}. We therefore do
not identify the stable family with Milburn's full Poisson model, nor do we
claim that every earlier intrinsic-decoherence model is Gaussian. The narrower
and rigorous statement is that a Gaussian integrated phase produces quadratic
charge dependence.

We also examine a bounded stochastic walk motivated by the fermionic reduction
models of Snoke and collaborators \cite{snoke2021,snoke2023,snoke2023entropy}.
For that discrete update, boundedness, symmetry, and nonzero conditional
variance imply almost-sure convergence to the two attractors and Born
probabilities, with no constraint on the shape of the proposal law. The theorem
does not extend to an untruncated $\alpha$-stable driver, whose jumps require a
separate stochastic-integral prescription. The two natural continuous-time
prescriptions are analyzed in section~\ref{rem:marcus} and neither reduces. In
the Gaussian diffusion limit the It\^o and exact-flow/Stratonovich readings are
inequivalent, he former is a bounded martingale that reduces, while the latter
becomes Brownian motion in a logit coordinate and has no asymptotic collapse.

Our results are organized in five layers. First, we characterize the scale-free
symmetric L\'evy exponents on $\R$ and derive the associated random-unitary
channel. Second, we distinguish the real lift from its wrapped circle process
and prove the exact complete-positivity boundary, including an explicit
three-charge obstruction above $\alpha=2$. Third, we establish the Gaussian
quadratic-charge theorem and give sufficient conditions for two non-Gaussian
mechanisms. Fourth, we prove the bounded-walk Born theorem and isolate the
stochastic-convention dependence. Fifth, we give an experimentally usable
scaling ratio together with an idealized calculation of the precision needed to
estimate the exponent. A five-dimensional geometric construction is discussed
separately in Section~\ref{sec:discussion_geometry} and is not used in any of
these results.

Throughout, three operations are kept distinct -- random-phase dephasing, which
is linear, trace preserving, and population preserving; mode filtering, which
is a trace-nonincreasing postselection operation; and nonlinear trajectory
selection, which is the only one that can produce an individual outcome.
Conflating them would turn dephasing into a claim about collapse that it cannot
support. Table~\ref{tab:three_ops} summarizes the distinction.

The paper is organized as follows. Section~\ref{sec:methods} gives the analytic
construction, the reduction theorem, the ensemble dynamics, and the numerical
protocols. Section~\ref{sec:results} reports the numerical checks and the
idealized precision calculation. Section~\ref{sec:discussion} relates the
quadratic endpoint to existing energy-dephasing models, separates established
conclusions from model assumptions, assesses the geometric proposal, and lists
falsification criteria. Technical proofs and bookkeeping appear in the
appendices.

\section{Methods}
\label{sec:methods}

\subsection{Lifted phase, wrapped phase, and the induced channel}
\label{sec:methods_channel}

For the channel statements, assume that $\hilb$ is finite dimensional. The same
formulas extend to a bounded self-adjoint $Q$ with discrete spectrum; unbounded
charges require the usual domain control and are not needed below. Let
\begin{equation}
  Q=\sum_a q_aP_a,\qquad P_aP_b=\delta_{ab}P_a,
  \qquad \sum_aP_a=\id,
  \label{eq:charge_operator}
\end{equation}
with real charges $q_a$. Let $(L_u)_{u\ge0}$ be a nontrivial, symmetric,
real-valued L\'evy process, independent of the system. Its characteristic
function can be written
\begin{equation}
  \E\!\left[e^{i\xi(L_u-L_0)}\right]=e^{-u\eta(\xi)},
  \label{eq:levy_exponent}
\end{equation}
where $\eta$ is continuous, even, conditionally negative definite,
$\eta(0)=0$, and $\eta\not\equiv0$. Symmetry makes
$e^{-u\eta(\xi)}=\E[\cos(\xi L_u)]$ real and bounded by one, so $\eta\ge0$.

The phrase ``no intrinsic charge scale'' will mean the following precise
covariance property: for every $\lambda>0$ there exists $c(\lambda)>0$ such that
\begin{equation}
  \eta(\lambda\xi)=c(\lambda)\eta(\xi)
  \quad\text{for all }\xi\in\R.
  \label{eq:scale_covariance}
\end{equation}
Rescaling the charge can then be absorbed into a rescaling of the semigroup
parameter. This condition is strictly stronger than scale invariance on the
integer Fourier modes of a circle, and is therefore imposed on the real lift.

\begin{theorem}[Scale-free symmetric L\'evy characterization]
\label{thm:scale_free}
Let $\eta$ satisfy the conditions above together with
Eq. \eqref{eq:scale_covariance}. Then there are unique constants $D>0$ and
$\alpha\in(0,2]$ such that
\begin{equation}
  \eta(\xi)=D\lvert\xi\rvert^{\alpha}.
  \label{eq:homogeneous_exponent}
\end{equation}
Conversely, every function in Eq. \eqref{eq:homogeneous_exponent} is a
symmetric L\'evy exponent and satisfies Eq. \eqref{eq:scale_covariance} with
$c(\lambda)=\lambda^{\alpha}$.
\end{theorem}

\begin{proof}
First, $\eta(1)>0$. Indeed $\eta\ge0$, and if $\eta(1)=0$ then
Eq. \eqref{eq:scale_covariance} gives $\eta(\lambda)=c(\lambda)\eta(1)=0$ for
every $\lambda>0$, so $\eta\equiv0$ by evenness and $\eta(0)=0$, contradicting
the hypothesis. Setting $\xi=1$ in Eq. \eqref{eq:scale_covariance} therefore
gives $c(\lambda)=\eta(\lambda)/\eta(1)$, which is positive and continuous.
Applying Eq. \eqref{eq:scale_covariance} twice yields
$c(\lambda\mu)=c(\lambda)c(\mu)$, and a positive continuous solution of the
multiplicative Cauchy equation is $c(\lambda)=\lambda^{\alpha}$ for a unique
real $\alpha$. Evenness and homogeneity of degree $\alpha$ then give
$\eta(\xi)=D\lvert\xi\rvert^{\alpha}$ with $D=\eta(1)>0$. Continuity at the
origin with $\eta(0)=0$ forces $\alpha>0$, and the L\'evy--Khintchine
characterization of continuous conditionally negative definite functions on
$\R$ restricts the degree of homogeneity to $\alpha\le2$. Conversely,
$D\lvert\xi\rvert^{\alpha}$ is the exponent of a symmetric stable process for
precisely this range \cite{sato1999}.
\end{proof}

The lifted random-unitary channel is
\begin{equation}
  \cE_u(\rho)
  =\E\!\left[e^{iL_uQ}\rho e^{-iL_uQ}\right].
  \label{eq:lifted_channel}
\end{equation}
It is a convex average of unitary conjugations, hence completely positive,
trace preserving and unital; stationarity and independence of the increments of
$L$ make $(\cE_u)_{u\ge0}$ a semigroup. Its charge blocks obey
\begin{equation}
  P_a\cE_u(\rho)P_b
  =e^{-Du\lvert q_a-q_b\rvert^{\alpha}}P_a\rho P_b.
  \label{eq:dephasing_blocks}
\end{equation}
Thus, populations are fixed and only coherences decay, giving
Eq. \eqref{eq:master_rate}.

Now define the wrapped process $\Theta_u=L_u\bmod 2\pi$. A circle-valued phase
acts consistently by conjugation when
\begin{equation}
  q_a-q_b\in\Z\quad\text{for all }a,b,
  \label{eq:charge_lattice}
\end{equation}
because shifting $\theta$ by $2\pi$ then changes $e^{i\theta Q}$ by at most a
common phase, which cancels in conjugation. Equivalently, the spectrum of $Q$
lies in a single affine lattice $\chi+\Z$. When Eq. \eqref{eq:charge_lattice}
fails, Eq. \eqref{eq:lifted_channel}, rather than an integral over a circle, is
the well-defined construction.

\begin{theorem}[Wrapped symmetric stable semigroup]
\label{thm:wrapped_stable}
For $0<\alpha\le2$ and $D>0$, the wrapped process has the transition density
\begin{equation}
  p_u(\theta)=\frac1{2\pi}\sum_{k\in\Z}
  e^{-Du\lvert k\rvert^{\alpha}}e^{ik\theta},
  \qquad u>0,
  \label{eq:transition_density}
\end{equation}
which defines a strongly continuous Markov convolution semigroup on $\Circ$.
Its generator is the closure of $-D(-\Delta_\theta)^{\alpha/2}$, its unique
invariant probability measure is normalized Haar measure, and
$p_u\to(2\pi)^{-1}$ in $L^2$ and weakly as $u\to\infty$.
\end{theorem}

A proof is given in Appendix~\ref{app:levy}. At $\alpha=1$ the Fourier series
sums to the Poisson kernel
\begin{equation}
  p_u(\theta)=\frac{1}{2\pi}
  \frac{1-r_u^2}{1-2r_u\cos\theta+r_u^2},
  \qquad r_u=e^{-Du},
  \label{eq:poisson_kernel}
\end{equation}
which is the wrapped Cauchy density \cite{mardia2000}. Wrapping removes any
literal tail, since a circle-valued variable has bounded support and hence
finite circular moments of every order. What survives is the spectral scaling
$\lvert k\rvert^{\alpha}$ of the Fourier modes, and it is that scaling, not any
tail, which is observable through Eq. \eqref{eq:master_rate}.

\begin{theorem}[Universal complete-positivity boundary]
\label{thm:cp_boundary}
For every finite set of real charges $\{q_a\}$ and every $u\ge0$, the Schur
multiplier
\begin{equation}
  C^{(u)}_{ab}=e^{-Du\lvert q_a-q_b\rvert^{\alpha}}
  \label{eq:schur_matrix}
\end{equation}
is positive semidefinite when $0<\alpha\le2$. For every $\alpha>2$ the
universal statement fails: the three charges $0,1,2$ give a non-positive
multiplier for all sufficiently small positive $u$.
\end{theorem}

\begin{proof}
For $0<\alpha\le2$, the kernel $\lvert x-y\rvert^{\alpha}$ is conditionally
negative definite on $\R$. Schoenberg's theorem therefore makes its negative
exponential positive definite, so Eq. \eqref{eq:schur_matrix} is positive
semidefinite for every finite charge set \cite{schoenberg1938}. For
$\alpha>2$, take $q=(0,1,2)$ and $v=(1,-2,1)$, for which $\sum_av_a=0$ and
\begin{equation}
  \sum_{a,b}v_av_b\lvert q_a-q_b\rvert^{\alpha}
  =2\bigl(2^{\alpha}-4\bigr)>0.
  \label{eq:three_charge_form}
\end{equation}
Since $v^{\mathsf T}\mathbf 1\mathbf 1^{\mathsf T}v=0$, expansion of
$C^{(u)}$ at $u=0$ gives
$v^{\mathsf T}C^{(u)}v=-Du\,2(2^{\alpha}-4)+O(u^2)<0$ for sufficiently small
$u$. Hence the proposed universal semigroup is not completely positive above
two.
\end{proof}

The boundary is exact and the matrix diagonalizations reported in
Section~\ref{sec:results_kernels} validate the implementation rather than
locate the boundary. Two consequences deserve emphasis. First, three charges
are the minimum needed. A two-level multiplier
$\bigl(\begin{smallmatrix}1&c\\c&1\end{smallmatrix}\bigr)$ with $c\in(0,1]$ is
positive semidefinite for every exponent, so a single two-level experiment
cannot by itself exhibit the obstruction. What fails above $\alpha=2$ is the
claimed scale-free characteristic function and its validity for arbitrary
charge spectra, not the positivity of every particular finite matrix. Second,
Theorem~\ref{thm:cp_boundary} asserts failure for sufficiently small $u$ at
three charges, whereas a broad charge set makes the failure visible over a wide
range of $u$; the analytic and numerical statements are therefore complementary
rather than redundant.

Phase dephasing should also be distinguished from a filter. If $P_k$ grades an
excitation level, then
\begin{equation}
  \cF_u(\rho)=M_u\rho M_u^\dagger,
  \qquad M_u=\sum_{k\ge0}e^{-\kappa k^{\beta}u}P_k,
  \label{eq:mode_filter}
\end{equation}
is completely positive and trace non-increasing. Provided
$\Tr(P_0\rho)>0$, its normalized output converges in trace norm to
$P_0\rho P_0/\Tr(P_0\rho)$. It is not a projector at finite $u$, and neither
$\cE_u$ nor $\cF_u$ selects an individual outcome.

\begin{table}[htbp]
\centering
\caption{Three distinct operations. Only nonlinear trajectory selection can
produce a definite individual outcome.}
\label{tab:three_ops}
\vspace{1mm}
\small
\begin{tabularx}{\textwidth}{>{\raggedright\arraybackslash}p{0.19\textwidth} X X X}
\toprule
& Phase dephasing $\cE_u$ & Mode filter $\cF_u$ & Nonlinear selection \\
\midrule
Type & CPTP and unital & CP, trace non-increasing & nonlinear on trajectories \\
Effect on populations & none & suppresses selected sectors & drives to attractors \\
Idempotent limit & no & yes, as $u\to\infty$ & not applicable \\
Produces outcomes & no & no & yes \\
Status here & derived & mathematical comparison & separately modeled \\
\bottomrule
\end{tabularx}
\end{table}

\subsection{Charge dependence from Gaussian and non-Gaussian phase models}
\label{sec:methods_exponent}

The linear coupling $e^{iL_uQ}$ is an assumption about how the unobserved phase
acts on the system. It is not a consequence of a Dirac rather than a
Klein--Gordon field equation, and none of the channel results selects fermions.
Fermionic occupation projectors enter only in the separate reduction model of
Section~\ref{sec:methods_reduction}.

A common microscopic parametrization writes the relative accumulated phase as
\begin{equation}
  \delta\phi_{ab}(u)=(q_a-q_b)X(u),
  \qquad X(u)=\int_0^u h(s)\dd s,
  \label{eq:integrated_phase}
\end{equation}
where $h$ is a centered stochastic field for which the integral is defined.
The charge dependence and the time dependence must be kept separate.

\begin{theorem}[Gaussian integrated phases are quadratic in charge]
\label{thm:gaussian_nogo}
Let $h$ be a centered Gaussian process with covariance
$C_h(s,t)=\E[h(s)h(t)]$, and let $X(u)$ be the linear functional in
Eq. \eqref{eq:integrated_phase}. Then
\begin{equation}
  \E\!\left[e^{i(q_a-q_b)X(u)}\right]
  =\exp\!\left[-\half(q_a-q_b)^2\Var X(u)\right],
  \qquad
  \Var X(u)=\int_0^u\!\int_0^u C_h(s,t)\dd s\dd t.
  \label{eq:gaussian_cf}
\end{equation}
Consequently Gaussianity fixes the charge exponent to two. The evolution is the
$\alpha=2$ semigroup of Eq. \eqref{eq:master_rate} only when
$\Var X(u)=2Du$. Otherwise its charge dependence is still quadratic but its
time dependence can be non-Markovian.
\end{theorem}

\begin{proof}
Every linear functional of a Gaussian process is Gaussian. The characteristic
function of a centered Gaussian variable of variance $v$ is
$e^{-\xi^2v/2}$. Substitute $\xi=q_a-q_b$. No assumption about the convergence
or truncation of a cumulant series is required.
\end{proof}

A Lorentzian feature in a frequency-domain spectrum should therefore not be
confused with a Cauchy increment distribution. The former is a property of
$C_h$ or of the spectral density.
The latter is a property of the full one-time
characteristic function. The two are logically independent, and Gaussianity
fixes the second at exponent two whatever the first.

An idealized inverse-power Poisson field supplies a non-Gaussian example under
explicit hypotheses. Let $\{X_j\}$ be a homogeneous Poisson point process of
intensity $n$ in $\R^d$, and let $\{A_j\}$ be independent symmetric marks,
independent of the points. Formally set
\begin{equation}
  Y=\sum_j A_j\lvert X_j\rvert^{-p}.
  \label{eq:poisson_field}
\end{equation}
The random field is understood as the symmetric Poisson stochastic integral
determined by its characteristic functional, equivalently as a distributional
limit of annular sums under symmetric truncation. This qualification matters
when the far-field series is not absolutely convergent.

\begin{proposition}[Ideal inverse-power Poisson field]
\label{prop:poisson_stable}
Assume $p>d/2$, $\Pr(A_1\ne0)>0$, and
$\E\lvert A_1\rvert^{d/p}<\infty$. Then the infinite-volume characteristic
functional of Eq. \eqref{eq:poisson_field} is symmetric stable,
\begin{equation}
  \E[e^{i\xi Y}]=e^{-C\lvert\xi\rvert^{\alpha}},
  \qquad \alpha=\frac{d}{p}\in(0,2),
  \label{eq:alpha_dp}
\end{equation}
where
\begin{equation}
  C=\frac{n\lvert\mathbb S^{d-1}\rvert}{p}
    \E\lvert A_1\rvert^{\alpha}
    \int_0^\infty (1-\cos s)s^{-1-\alpha}\dd s>0.
  \label{eq:poisson_constant}
\end{equation}
\end{proposition}

The proof is given in Appendix~\ref{app:influence}. Three points about the
hypotheses are worth stating explicitly, because they determine how much
physical freedom the mechanism actually has.

First, the convergence condition $p>d/2$ comes from the \emph{far} field, where
$1-\cos(\xi Ar^{-p})\simeq\half\xi^2A^2r^{-2p}$ and the radial integrand behaves
as $r^{d-1-2p}$. It is not an extra assumption: $p>d/2$ is equivalent to
$\alpha=d/p<2$, so the stable window is exactly the convergence window. At the
opposite end, $p\le d/2$ makes the centered sum square integrable and the
central limit theorem returns a Gaussian limit, that is, the boundary value
$\alpha=2$.

Second, no small-$r$ regularization is needed for the characteristic function
to exist. The integrand $1-\cos(\cdot)$ is bounded by two and $r^{d-1}$ is
integrable at the origin, so Eq. \eqref{eq:poisson_log_cf} converges there
without a cutoff, even though ordinary moments of the singular field diverge. A
hard core does regularize those moments, but it also removes the exact stable
tail, so a finite system or a regularized source yields at most an approximate
stable regime.

Third, the exponent is fixed by the falloff of the microscopic coupling rather
than freely adjustable. In $d=3$ the window admits a force-type coupling with
$p=2$, giving the Holtsmark value $\alpha=3/2$
\cite{holtsmark1919,chandrasekhar1942,chandrasekhar1943}, and it admits
$d=p=3$, giving a Cauchy field.
This last case is the classical mechanism
behind Lorentzian pressure-broadened spectral lines, in which a perturber
shifts the transition by an amount falling as the inverse cube of the impact
parameter \cite{chandrasekhar1943}. A Coulomb potential-type coupling with
$p=1$ lies outside the window, since $\alpha=3$, and the infinite-volume
characteristic function does not converge.

A second mechanism is subordination. Let $B_s$ be Brownian phase with
\begin{equation}
  \E[e^{i\xi B_s}]=e^{-\kappa\xi^2s},
  \label{eq:brownian_phase}
\end{equation}
and let $S_u$ be an independent $\beta$-stable subordinator with
\begin{equation}
  \E[e^{-\lambda S_u}]=e^{-cu\lambda^{\beta}},
  \qquad 0<\beta\le1.
  \label{eq:subordinator}
\end{equation}
Then
\begin{equation}
  \E[e^{i\xi B_{S_u}}]
  =e^{-c\kappa^{\beta}u\lvert\xi\rvert^{2\beta}}
  =e^{-Du\lvert\xi\rvert^{\alpha}},
  \qquad \alpha=2\beta,
  \quad D=c\kappa^{\beta}.
  \label{eq:subordination_cf}
\end{equation}
The Cauchy process is obtained from $\beta=1/2$, and the full stable range from
$\beta=\alpha/2$. The randomness is transferred from the Gaussian phase to the
operational clock rather than inserted as a non-Gaussian phase amplitude.

Neither mechanism fixes the dimensional rate $D$ for a laboratory system.
Therefore, equation \eqref{eq:master_rate} should be treated as a two-parameter
phenomenology, with $D$ and $\alpha$ to be constrained or estimated. Under the
energy assignment $q_a=E_a/\Estar$ with a reference energy $\Estar$, the
measurable coefficient is $\Lambda_\alpha=D/\Estar^{\alpha}$, with units
$\mathrm{s}^{-1}\mathrm{energy}^{-\alpha}$.

\subsection{A bounded reduction walk and the Born rule}
\label{sec:methods_reduction}

In the norm-conserving reduction models discussed by Snoke and collaborators
\cite{snoke2023,snoke2023entropy},
a state-dependent anti-Hermitian term is written
\begin{equation}
  \widehat V
  =\sum_a i\h\,\omega_{R,a}
     \bigl(\langle\widehat N_a\rangle-\widehat N_a\bigr),
  \label{eq:snoke_V}
\end{equation}
where $\widehat N_a$ is an occupation projector,
$\widehat N_a^2=\widehat N_a$, and $\omega_{R,a}$ fluctuates with the
environment \cite{snoke2021,snoke2023,snoke2023entropy}. Relative to one mode,
a state can be decomposed as
\begin{equation}
  a\lvert\psi_0\rangle\lvert0\rangle
  +b\lvert\psi_1\rangle\lvert1\rangle,
  \qquad \langle\psi_0\vert\psi_1\rangle=0,
  \label{eq:mode_decomposition}
\end{equation}
with $U_3(0)=2\lvert b\rvert^2-1$. On the invariant plane $U_1=0$, the formal
Bloch equations used in that model reduce to
\begin{equation}
  \dot U_2=-\omega_RU_3(1-U_3^2)^{1/2},
  \qquad
  \dot U_3=\omega_RU_2(1-U_3^2)^{1/2},
  \label{eq:bloch_reduction}
\end{equation}
with fixed points $U_3=\pm1$.

The analytic result used for the simulations is most cleanly stated directly
for the bounded discrete update
\begin{equation}
  U_{n+1}=U_n+\eps\varphi_{n+1}(1-U_n^2).
  \label{eq:bounded_walk}
\end{equation}
Let $\mathcal F_n=\sigma(U_0,\varphi_1,\ldots,\varphi_n)$ be the natural
filtration of the driver.

\begin{proposition}[Born rule for a bounded symmetric walk]
\label{prop:martingale}
Suppose $U_0\in[-1,1]$, the variables $\varphi_n$ are independent and
identically distributed, symmetric, bounded by $\lvert\varphi_n\rvert\le M$,
and have variance $v_\varphi>0$. If $0<\eps M\le1/2$, then
Eq. \eqref{eq:bounded_walk} remains in $[-1,1]$, is a bounded martingale, and
converges almost surely to $U_\infty\in\{-1,+1\}$. Moreover
\begin{equation}
  \Pr(U_\infty=+1)=\frac{1+U_0}{2}.
  \label{eq:born_walk}
\end{equation}
\end{proposition}

\begin{proof}
Write $\delta=\eps\varphi_{n+1}\in[-1/2,1/2]$ and
$f(U)=U+\delta(1-U^2)$. Then $f'(U)=1-2\delta U\ge0$ on $[-1,1]$ and
$f(\pm1)=\pm1$, so $f$ maps $[-1,1]$ into itself and the update is state
preserving. Symmetry and independence give
$\E[U_{n+1}\mid\mathcal F_n]=U_n$, so $(U_n)$ is a bounded, hence uniformly
integrable, martingale, and it converges almost surely and in $L^1$ to a limit
$U_\infty$. In addition
\begin{equation}
  \E[U_{n+1}^2\mid\mathcal F_n]
  =U_n^2+\eps^2v_\varphi(1-U_n^2)^2 .
  \label{eq:second_moment_recursion}
\end{equation}
Taking expectations and telescoping, and using $U_n^2\le1$, gives
$\sum_n\E[(1-U_n^2)^2]\le(1-U_0^2)/(\eps^2v_\varphi)<\infty$. Markov's
inequality and the Borel--Cantelli lemma then imply $1-U_n^2\to0$ almost
surely, so $U_\infty=\pm1$. Finally $\E[U_\infty]=U_0$ by $L^1$ convergence,
which yields Eq. \eqref{eq:born_walk}.
\end{proof}

\begin{corollary}
\label{cor:alpha_free}
For any family of bounded symmetric proposal laws satisfying the hypotheses of
Proposition~\ref{prop:martingale}, in particular the symmetrically truncated
stable laws used in Section~\ref{sec:results_born}, the branch probability is
the same function of $U_0$ and is independent of the parent stability index
$\alpha$. Born-rule recovery in this model class therefore places no constraint
on the exponent.
\end{corollary}

Corollary~\ref{cor:alpha_free} matters for the logic of the paper. If the
exponent were fixed by the requirement of Born-rule recovery, it would not be
an empirical quantity. Because it is not fixed, it becomes a measurement
target. Numerical studies of this model class report the same branch statistics
under Gaussian and under Lorentzian driving \cite{snoke2023entropy}, which is
what Proposition~\ref{prop:martingale} predicts.

The hypotheses of Proposition~\ref{prop:martingale} are important. After
truncation the proposal law is not itself stable, and the truncation cannot be
removed simply by the boundedness of the state. The factor $(1-U_n^2)$ in
Eq. \eqref{eq:bounded_walk} is $\mathcal F_n$-measurable and multiplies
$\varphi_{n+1}$, so it confers no integrability on the driver:
$\E\lvert\Delta U_n\rvert=\eps(1-U_n^2)\E\lvert\varphi\rvert$, which is
infinite whenever $\E\lvert\varphi\rvert$ is. An untruncated stable variable
with $\alpha\le1$ has no finite first absolute moment, so the conditional
expectation invoked above is undefined rather than merely unproved; and for
$1<\alpha<2$ the mean exists but the update is not state preserving because the
proposal is unbounded.

\subsubsection{Continuous-time stable drivers}\label{rem:marcus}
The two natural continuous-time prescriptions can be examined directly, and
neither reproduces Proposition~\ref{prop:martingale}.

Under the exact-flow, or Marcus canonical, equation the flow in the logit
coordinate $z=\operatorname{arctanh}U$ is the translation $z\mapsto z+\Delta L$,
so that $U_u=\tanh(z_0+L_u)$. A nondegenerate symmetric stable process on $\R$
oscillates, with $\limsup_{u\to\infty}L_u=+\infty$ and
$\liminf_{u\to\infty}L_u=-\infty$ almost surely, for every $\alpha\in(0,2]$.
Hence $U_u$ has no almost-sure limit and there is no reduction at any stability
index. This extends Proposition~\ref{prop:convention}(ii) from the Gaussian
case to the entire stable family.

Under the naive It\^o jump equation
$\dd U_u=\sigma(1-U_{u-}^2)\dd L_u$, the unbounded jumps carry the state
outside $[-1,1]$, so the equation is not state preserving; and for
$\alpha\le1$ the increment is not integrable, so the martingale property is
undefined. A well-posed continuous-time jump model therefore requires a
compensated, state-preserving jump measure, of which
Eq. \eqref{eq:bounded_walk} is the discrete-time representative.

The stochastic-calculus convention can be isolated without invoking stable
jumps at all, by considering Gaussian white noise.

\begin{proposition}[Gaussian convention dependence]
\label{prop:convention}
Let $U_0\in(-1,1)$.
\begin{enumerate}
\item[(i)] The It\^o diffusion
\begin{equation}
  \dd U_u=\sigma(1-U_u^2)\dd W_u
  \label{eq:ito_U}
\end{equation}
is a bounded martingale that converges to $\pm1$ and satisfies
$\Pr(U_\infty=+1)=(1+U_0)/2$.
\item[(ii)] The exact-flow or Wong--Zakai white-noise limit is the
Stratonovich equation
\begin{equation}
  \dd U_u=\sigma(1-U_u^2)\circ\dd W_u.
  \label{eq:strat_U}
\end{equation}
With $z=\operatorname{arctanh}U$ it becomes $\dd z=\sigma\dd W$, so it has no
almost-sure limit at $\pm1$. If reduction is declared on first reaching
$\lvert z\rvert=L$, equivalently $\lvert U\rvert=\tanh L$, then
\begin{equation}
  \Pr(\text{exit at }+L)
  =\half+\frac{\operatorname{arctanh}U_0}{2L},
  \label{eq:threshold_exit}
\end{equation}
which is not the Born rule and tends to $1/2$ as the threshold is moved toward
the formal attractors. Here $L$ is an explicit infrared cutoff on the logit
coordinate, that is, a declared distance from the attractors at which
absorption is deemed to have occurred.
\end{enumerate}
\end{proposition}

The proof, and the It\^o drift form of Eq. \eqref{eq:strat_U}, are given in
Appendix~\ref{app:convention}. The convention cannot be selected merely by
observing that the noise is external. It must follow from the microscopic
limiting procedure.
Ideal non-anticipating white-noise forcing leads naturally
to an It\^o model, whereas smooth colored approximations can lead to a
Stratonovich/Wong--Zakai limit. For jump noise, the analogous exact-flow
convention is Marcus rather than Stratonovich, as in section \ref{rem:marcus}. A
reduction model in this class must therefore state which convention it adopts,
because the two give quantitatively different and experimentally
distinguishable predictions.

\subsection{Ensemble linearity of a diffusive selection model}
\label{sec:methods_ensemble}

A standard normalized diffusive stochastic Schr\"odinger equation with bounded,
self-adjoint operators $N_a$ is
\begin{equation}
  \dd\lvert\Psi\rangle
  =-\frac{i}{\h}H\lvert\Psi\rangle\dd u
   +\sum_a\left[
     \gamma_a(N_a-\langle N_a\rangle)\dd W_u^a
     -\frac{\gamma_a^2}{2}(N_a-\langle N_a\rangle)^2\dd u
   \right]\lvert\Psi\rangle,
  \label{eq:collapse_sse}
\end{equation}
where $\dd W_u^a\dd W_u^b=\delta^{ab}\dd u$. That last identity is a property
of Wiener quadratic variation and does not hold for a general driver. It\^o's
product rule gives the closed ensemble equation
\begin{equation}
  \frac{\dd\rho}{\dd u}
  =-\frac{i}{\h}[H,\rho]
   -\half\sum_a\gamma_a^2[N_a,[N_a,\rho]],
  \qquad
  \rho=\E\lvert\Psi\rangle\langle\Psi\rvert.
  \label{eq:collapse_lindblad}
\end{equation}
The cancellation is derived in Appendix~\ref{app:nosignaling}. It uses the
Wiener quadratic variation for this particular diffusive unraveling, so a naive
replacement of $\dd W$ by an $\alpha$-stable increment does not yield a valid
stochastic Schr\"odinger equation. This does not imply that non-Gaussian or
jump unravelings are impossible. Poisson-driven quantum-jump unravelings, for
instance, have exactly linear Lindblad ensembles; such models simply require
their own L\'evy--It\^o compensators and must be checked separately for norm
preservation and ensemble linearity.

Linearity of Eq. \eqref{eq:collapse_lindblad} is necessary but not by itself
sufficient for operational no-signaling. In a bipartite application we also
assume that controllable operations are local and that any self-adjoint
dephasing operators are local or additive sums
$A_r\otimes\id+\id\otimes B_r$. Under these assumptions the reduced generator
on either party is closed and independent of the other party's local setting, see
Appendix~\ref{app:nosignaling} for the partial-trace calculation. A genuinely
nonlocal Hamiltonian, or a product Lindblad operator such as $A\otimes B$, can
mediate interactions and lies outside that argument.

The order in which the limit and the average are taken is not a formality. The
infinitesimal generator is defined through the conditional limit
\begin{equation}
  \cL(\Pi_u)
  =\lim_{\Delta u\downarrow0}
    \frac{\E[\Pi_{u+\Delta u}-\Pi_u\mid\mathcal F_u]}{\Delta u},
  \qquad \Pi_u=\lvert\Psi_u\rangle\langle\Psi_u\rvert,
  \label{eq:conditional_generator}
\end{equation}
and the quadratic-variation term must be included before the conditional
expectation is evaluated. Averaging over increments at finite step size, before
the limit is taken, discards precisely the quadratic-variation contribution
that cancels the state-dependent terms in Eq. \eqref{eq:ito_cancellation}. The
resulting finite-difference object is not the It\^o generator. It retains
residual $\langle N_a\rangle\rho$ terms in the marginal, and those terms present
themselves as signaling. Apparent signaling reported for models of this class
\cite{snoke2023entropy} is attributable to that order of operations. With the
limit taken first, Eq. \eqref{eq:collapse_lindblad} is exactly linear and, under
the locality assumptions above, the marginal is closed.

Finally, Eq. \eqref{eq:dephasing_blocks} preserves the expectation of every
observable $A$ commuting with $Q$
\begin{equation}
  \Tr[A\cE_u(\rho)]=\Tr(A\rho),\qquad [A,Q]=0.
  \label{eq:charge_diagonal_conservation}
\end{equation}
Energy conservation for the dephasing channel is therefore conditional, it
follows when $[H,Q]=0$, and in particular under the energy assignment
$Q=f(H)$. It is not a claim for an arbitrary charge operator.

The corresponding statement for the selection model is different and weaker.
Using cyclicity, the generator in Eq. \eqref{eq:collapse_lindblad} satisfies
\begin{equation}
  \frac{\dd}{\dd u}\Tr(H\rho)
  =-\half\sum_a\gamma_a^2\Tr\bigl([H,N_a][N_a,\rho]\bigr),
  \label{eq:selection_energy}
\end{equation}
which vanishes for all $\rho$ when $[H,N_a]=0$ for every $a$, and in general
does not vanish otherwise. A selection channel whose operators fail to commute
with the Hamiltonian injects energy, as continuous spontaneous localization
does. Exact energy conservation is thus a property of the dephasing channel
under $[H,Q]=0$, not of every model discussed here, and any comparison with
heating searches must specify which of the two is meant.

\subsection{Numerical protocols}
\label{sec:methods_numerics}

All numerical values reported below are generated by the supplementary
computational notebook with fixed random seeds. The analytic theorems do not
depend on these computations.

Transition kernels are evaluated from Eq. \eqref{eq:transition_density} with
Fourier modes $\lvert k\rvert\le4000$, far beyond convergence for the values of
$Du$ used. They are compared with the Poisson kernel at $\alpha=1$ and with
$4\times10^5$ lifted symmetric stable variates wrapped modulo $2\pi$. The
lifted scale is $(Du)^{1/\alpha}$, so that the characteristic function agrees
with Eq. \eqref{eq:stable_cf_intro} at integer arguments. The
complete-positivity implementation is checked by diagonalizing
Eq. \eqref{eq:schur_matrix} for the $64$ equally spaced charges
$q=0,\ldots,63$ and for $30$ real charges sampled uniformly from $[-5,5]$, on
a grid in $\alpha$ with spacing $5\times10^{-4}$ near two, at several values
of $Du$.

Born-rule trials use Eq. \eqref{eq:bounded_walk} with $\eps=0.35$ and a
symmetric stable parent law of unit scale conditioned on
$\lvert\varphi\rvert\le1$, so that $\eps M=0.35<1/2$ and the update remains in
$[-1,1]$ as required by Proposition~\ref{prop:martingale}. Symmetric truncation
preserves symmetry, and hence the martingale property, exactly. For each
$U_0\in\{-0.8,-0.4,0,0.4,0.8\}$ and $\alpha\in\{0.5,1,1.5,2\}$,
$2\times10^4$ independent trajectories are run until first entry into
$\lvert U\rvert\ge1-10^{-8}$, or for at most $6\times10^4$ steps. No trajectory
remained censored at this cap in the reported runs. The threshold changes
stopping times but changes the ideal symmetric exit probability only by
$O(10^{-8})$. The exact-flow control uses the coordinate
$z=\operatorname{arctanh}U$, the update $z_{n+1}=z_n+\eps\varphi_{n+1}$, a
truncated Gaussian proposal, $2\times10^4$ trajectories per initial value, and
a fixed horizon of $4000$ steps. Its reported quantity is the terminal branch
fraction $\Pr(U_{4000}>0)$, not an absorption probability.

For the subordination check, $5\times10^5$ samples are generated as
$X=\sigma\sqrt{S}\,Z$, where $Z$ is standard normal and $S$ is a positive
$1/2$-stable variable parameterized by
$\E[e^{-\lambda S}]=e^{-u\sqrt\lambda}$. Hence,
\begin{equation}
  \E[e^{i\xi X}]=e^{-(\sigma/\sqrt2)u\lvert\xi\rvert},
  \label{eq:subordination_target}
\end{equation}
and the real empirical characteristic function, that is the sample mean of
$\cos(\xi X)$, is compared directly with this expression.

The idealized exponent-recovery study uses splittings proportional to
$\{1,2\}$, $\{1,2,4\}$, and $\{1,2,5,10\}$. Synthetic visibilities are
\begin{equation}
  V_i=\exp[-A x_i^{\alpha}]+\mathcal N(0,\sigma_V^2),
  \qquad A=-\log(0.95),
  \label{eq:precision_model}
\end{equation}
so that the smallest splitting loses five percent of its visibility. Each
replicate is fitted by constrained nonlinear least squares on the raw
visibilities, with $A>0$ and $0.05\le\alpha\le3$. Fitting the raw visibilities
rather than a doubly logarithmic transform is essential here: with additive
readout noise the realized $V_i$ can leave $(0,1)$, at which point
$\log(-\log V_i)$ is undefined, and discarding or censoring the affected
replicates biases the estimator badly in exactly the configurations with the
longest lever arm. The standard deviation of $\widehat\alpha$ is estimated from
$4000$ replicates using random seed 20260818. The calculation assumes known
initial visibility, independent homoscedastic readout noise, and no nuisance
environmental channel; it is a design benchmark rather than a forecast for a
specific apparatus.

\section{Results}
\label{sec:results}

\subsection{Wrapped kernels and a numerical check of complete positivity}
\label{sec:results_kernels}

Direct evaluation of Eq. \eqref{eq:transition_density} agrees with the closed
Poisson-kernel form, Eq. \eqref{eq:poisson_kernel}, to maximum absolute errors
$6.2\times10^{-14}$ at $Du=0.1$ and $8.3\times10^{-17}$ at $Du=2$. These are
floating-point truncation errors rather than statistical discrepancies.

The wrapping operation was checked independently by drawing $4\times10^5$
lifted symmetric stable variates. For circular moments $k=1,\ldots,5$ the
maximum deviations from $e^{-Dk^{\alpha}u}$ were $1.8\times10^{-3}$ at
$\alpha=0.7$, $1.2\times10^{-3}$ at $\alpha=1$, $1.4\times10^{-3}$ at
$\alpha=1.5$, and $8.9\times10^{-4}$ at $\alpha=2$, against a representative
Monte Carlo standard error of $1.6\times10^{-3}$. Every deviation lies within
$1.2$ standard errors. Figure~\ref{fig:kernels}(a) displays the kernels; as
$\alpha$ decreases the density develops the sharp central peak and broad
shoulders characteristic of heavy-tailed increments while remaining an ordinary
bounded density on the circle, which is the concrete form of the remark
following Eq. \eqref{eq:poisson_kernel}.

\begin{figure}[htbp]
\centering
\journalfigure[width=\textwidth]{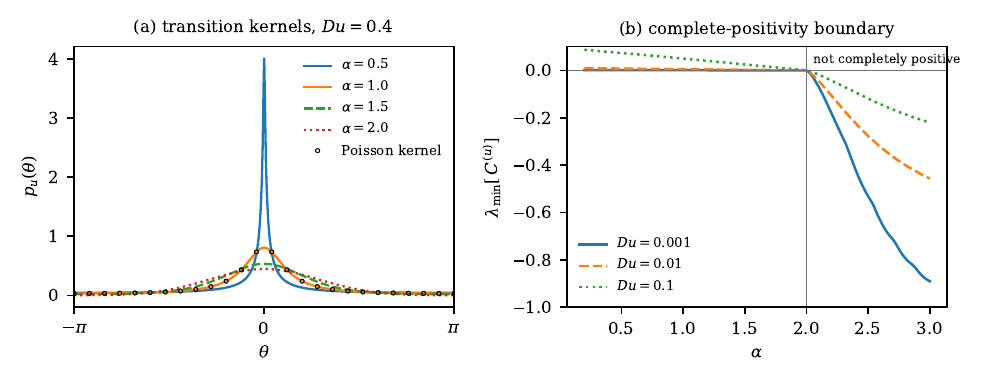}
\caption{(a) Wrapped stable transition densities at $Du=0.4$ for
$\alpha=0.5,1,1.5,2$, computed from Eq. \eqref{eq:transition_density}.
Open circles show the Poisson kernel of Eq. \eqref{eq:poisson_kernel}, which
coincides with the $\alpha=1$ curve to machine precision.
(b) Floating-point smallest eigenvalue of the Schur matrix
$C^{(u)}_{ab}=e^{-Du\lvert q_a-q_b\rvert^{\alpha}}$ for $q=0,\ldots,63$. The
exact universal boundary is proved analytically in
Theorem~\ref{thm:cp_boundary}; the scan shows how negativity appears in this
finite charge set above two. Near $\alpha=2$ and small $Du$ the matrix is
severely ill-conditioned, and eigenvalues at the level of machine precision
should be read as numerically, not exactly, zero.}
\label{fig:kernels}
\end{figure}

For $q=0,\ldots,63$ and $Du=10^{-3}$, the computed smallest eigenvalue at
$\alpha=2$ is $-7\times10^{-15}$, which is roundoff in an exactly positive
Gaussian kernel. With a grid spacing of $5\times10^{-4}$, the first negative
value occurs at $\alpha=2.0005$ for $Du=10^{-3},10^{-2},10^{-1}$ alike, and at
$\alpha=3$ the smallest eigenvalue is $-0.89$ for $Du=10^{-3}$. The same
behavior holds for generic real charges: for thirty values drawn uniformly from
$[-5,5]$ at $Du=1$, the smallest eigenvalue is positive at $\alpha=1$,
indistinguishable from zero at $\alpha=2$, and negative at $\alpha=2.05$.
Figure~\ref{fig:kernels}(b) displays the transition. These values validate the
code path across a range of $u$ that the small-$u$ expansion in
Theorem~\ref{thm:cp_boundary} does not itself cover; they do not improve the
exact analytic boundary.

Therefore, a globally scale-free rate proportional to $\lvert\Delta q\rvert^{\alpha}$ with
$\alpha>2$ cannot be the characteristic exponent of a symmetric
random phase on $\R$, and cannot define a completely positive trace-preserving
semigroup for arbitrary charge spectra. This statement is narrower than the
assertion that every effective two-level fit with a local slope above two is
nonphysical. Finite-range crossovers, and more general phase distributions that
are not pure power laws, are not excluded by Theorem~\ref{thm:cp_boundary}.

\subsection{Born statistics for the bounded walk}
\label{sec:results_born}

Table~\ref{tab:born} reports the first-threshold branch fraction for the
bounded update of Eq. \eqref{eq:bounded_walk}. For each parent index and each
initial value, $2\times10^4$ independent trajectories were simulated. All
twenty branch fractions agree with the Born prediction $\half(1+U_0)$ within
the expected Monte Carlo scatter. The largest absolute deviation is $0.0092$,
the largest standardized deviation is $2.8$, which is unremarkable across
twenty comparisons, and the fitted slopes and intercepts are close to $1/2$, as
required by Proposition~\ref{prop:martingale}. The Cauchy and Gaussian cases
are statistically indistinguishable from each other and from the two
intermediate indices. Figure~\ref{fig:born}(a) displays the collapse onto the
Born line. The labels $\alpha=0.5,1,1.5,2$ refer to the untruncated parent
distributions; the simulated proposals are their bounded conditional laws.

\begin{table}[htbp]
\centering
\caption{Born-rule recovery for four bounded symmetric proposal laws. The first
four rows report the fraction of $2\times10^4$ trajectories first entering the
upper threshold $U\ge1-10^{-8}$; the Born prediction is $\half(1+U_0)$. Monte
Carlo standard errors are at most $0.0035$. The final row is a different
quantity: it reports the terminal fraction $\Pr(U_{4000}>0)$ under the
exact-flow Gaussian control and is not an absorption probability. Here
$\widehat a$ and $\widehat b$ are the slope and intercept of a fit against
$U_0$.}
\label{tab:born}
\vspace{1mm}
\begingroup
\scriptsize
\setlength{\tabcolsep}{3pt}
\begin{tabular}{@{}lccccccccc@{}}
\toprule
& \multicolumn{5}{c}{$U_0$} & & & & \\
\cmidrule(lr){2-6}
Proposal law & $-0.8$ & $-0.4$ & $0$ & $0.4$ & $0.8$
 & $\max\lvert\mathrm{dev}\rvert$ & $\widehat a$ & $\widehat b$ & $R^2$\\
\midrule
Born prediction & 0.100 & 0.300 & 0.500 & 0.700 & 0.900
  & --- & 0.5 & 0.5 & 1\\
\midrule
parent $\alpha=0.5$          & 0.1027 & 0.2989 & 0.5044 & 0.6908 & 0.8999 & 0.0092 & 0.4966 & 0.4993 & 0.99977\\
parent $\alpha=1$ (Cauchy)   & 0.1007 & 0.3040 & 0.5017 & 0.6990 & 0.9023 & 0.0040 & 0.4996 & 0.5015 & 0.99997\\
parent $\alpha=1.5$          & 0.0977 & 0.3003 & 0.5022 & 0.7004 & 0.8994 & 0.0023 & 0.5009 & 0.5000 & 0.99998\\
parent $\alpha=2$ (Gaussian) & 0.0961 & 0.2973 & 0.5020 & 0.6994 & 0.8970 & 0.0039 & 0.5010 & 0.4984 & 0.99995\\
\midrule
Exact-flow control at step 4000
 & 0.4614 & 0.4844 & 0.4994 & 0.5124 & 0.5387 & 0.3614 & 0.0456 & 0.4993 & 0.98577\\
\bottomrule
\end{tabular}
\endgroup
\end{table}

The exact-flow control behaves entirely differently, and the contrast is the
substantive content of this section. Its branch fraction varies only between
$0.4614$ and $0.5387$ as $U_0$ ranges over $[-0.8,0.8]$, against the required
range $[0.1,0.9]$; the fitted slope is $0.046$ rather than $0.5$, a factor of
eleven too small. This is the numerical signature of the drift toward the
equator identified in Proposition~\ref{prop:convention}(ii), and the agreement
with the analytic prediction is quantitative. With a standard normal proposal
truncated to $[-1,1]$ the proposal variance is $0.291$, so the accumulated $z$
walk has standard deviation
\begin{equation}
  s_{4000}=0.35\sqrt{0.291\times4000}=11.9 .
  \label{eq:accumulated_sd}
\end{equation}
Since $U=\tanh z$, the central-limit normal approximation for the accumulated
truncated-Gaussian increments predicts
\begin{equation}
  \Pr(U_{4000}>0)
  \simeq\Phi\!\left(\frac{\operatorname{arctanh}U_0}{11.9}\right),
  \label{eq:finite_horizon_control}
\end{equation}
giving $0.4634$, $0.4859$, $0.5000$, $0.5141$, and $0.5366$ against the
observed $0.4614$, $0.4844$, $0.4994$, $0.5124$, and $0.5387$. Every difference
is below $0.6$ Monte Carlo standard errors. Equation \eqref{eq:finite_horizon_control}
is a fixed-horizon sign calculation and should not be described as gambler's
ruin; the gambler's-ruin formula is the threshold-exit probability of
Eq. \eqref{eq:threshold_exit}, which is linear in $\operatorname{arctanh}U_0$
rather than a normal distribution function.

\begin{figure}[htbp]
\centering
\journalfigure[width=\textwidth]{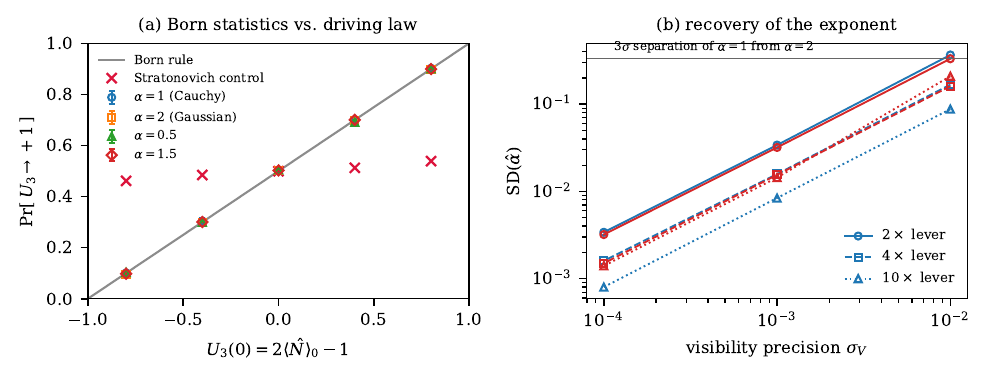}
\caption{(a) Branch probability against initial occupation for four bounded
symmetric proposal laws, with $3\sigma$ Monte Carlo error bars. All four
collapse onto the Born line $\half(1+U_0)$. Crosses show the exact-flow
Gaussian control at step $4000$, which is nearly flat at $\half$.
(b) Standard deviation of the fitted exponent $\widehat\alpha$ against
absolute visibility noise $\sigma_V$, for splitting sets spanning factors of
$2$, $4$, and $10$; blue denotes $\alpha=1$ and red $\alpha=2$. The horizontal
line marks the precision at which $\alpha=1$ and $\alpha=2$ separate at three
standard deviations.}
\label{fig:born}
\end{figure}

Mean absorption times differ across the proposal laws, being $253$, $281$,
$263$, and $421$ steps at $\alpha=1.5,1,2,0.5$ respectively for the symmetric
initial condition. This ordering is not a statement about heavy tails.
Truncating a heavier-tailed law at unit scale concentrates more of its mass
near the origin and so reduces the effective step variance, and the ordering
tracks that variance rather than the tail index. Absorption time carries no
invariant meaning here and is not used.

Two separate conclusions follow. Within the bounded non-anticipating walk, the
Born probability does not select a stable parent index, consistent with
Corollary~\ref{cor:alpha_free}. Under the exact-flow Gaussian convention the
state does not asymptotically reduce, and any imposed finite threshold obeys
the non-Born exit law of Proposition~\ref{prop:convention}.

\emph{Subordination check.}\label{sec:results_subord}
For the $1/2$-stable clock parameterization of
Section~\ref{sec:methods_numerics}, the empirical characteristic function from
$5\times10^5$ samples was $0.83698$, $0.70175$, $0.49260$, $0.24274$, and
$0.05981$ at Fourier arguments $0.25,0.5,1,2,4$. The corresponding Cauchy
values are $0.83797$, $0.70219$, $0.49307$, $0.24312$, and $0.05911$. The
largest absolute difference, $9.9\times10^{-4}$, is below the quoted Monte
Carlo standard error of $1.4\times10^{-3}$. A linear fit of the logarithm of
the empirical characteristic function against $\lvert\xi\rvert$ gives
$D=0.7038$, compared with the analytic value $\sigma/\sqrt2=0.7071$, a relative
difference of $0.5\%$. The reported metric is thus a direct fit of
$\log\E[e^{i\xi X}]$ against $\lvert\xi\rvert^{\alpha}$.
It checks the
characteristic-function identity used in Eq. \eqref{eq:subordination_cf} and is
not a Kolmogorov--Smirnov tail test.

\subsection{Estimating the exponent from rate scaling}
\label{sec:results_precision}

Under the phenomenological energy assignment $q_a=E_a/\Estar$, a two-level
coherence with splitting $\Delta E$ has anomalous visibility
\begin{equation}
  V(T)=V_0\exp[-\Lambda_\alpha\lvert\Delta E\rvert^{\alpha}T].
  \label{eq:visibility_law}
\end{equation}
If the anomalous rate can be isolated at two splittings, its unknown overall
coefficient cancels
\begin{equation}
  R_k=\frac{\Gamma(k\Delta E)}{\Gamma(\Delta E)}=k^{\alpha}.
  \label{eq:scaling_ratio}
\end{equation}
At $k=2$, a splitting-independent phenomenology predicts $R_2=1$, the Cauchy
hypothesis predicts $R_2=2$, and the Gaussian hypothesis predicts $R_2=4$.
Equation \eqref{eq:scaling_ratio} is the cleanest scale-free discriminator the
framework offers, but it applies to the isolated anomalous component and not
automatically to the total laboratory decoherence rate.

Table~\ref{tab:precision} reports the standard deviation of the fitted exponent
under the protocol of Section~\ref{sec:methods_numerics}. With two points a
factor of two apart and one-percent absolute visibility noise, the standard
deviation of $\widehat\alpha$ is about $0.32$--$0.35$, so $\alpha=1$ and
$\alpha=2$ separate by only about three standard deviations in this idealized
model. Three points spanning a factor of four reduce the uncertainty to
$0.08$--$0.13$, and four points spanning a factor of ten reduce it to about
$0.05$--$0.06$. At smaller readout noise the uncertainty scales nearly linearly
with $\sigma_V$. Lever arm is therefore worth as much as raw visibility
precision, and often more. Figure~\ref{fig:born}(b) shows the scaling of the
estimator's standard deviation with both.

\begin{table}[htbp]
\centering
\caption{Standard deviation of the constrained nonlinear least-squares estimate
$\widehat\alpha$ from $4000$ synthetic raw-visibility data sets per
configuration. Splittings are given in units of the smallest splitting, whose
true visibility is $0.95$. The calculation assumes known $V_0$, independent
additive Gaussian readout noise, and no background-decoherence nuisance
parameters.}
\label{tab:precision}
\vspace{1mm}
\small
\begin{tabular}{llccc}
\toprule
& & \multicolumn{3}{c}{absolute visibility noise $\sigma_V$}\\
\cmidrule(lr){3-5}
Splittings & True $\alpha$ & $10^{-2}$ & $10^{-3}$ & $10^{-4}$\\
\midrule
$\{1,2\}$       & 1 & 0.350 & 0.0337 & 0.00340\\
                 & 2 & 0.322 & 0.0309 & 0.00305\\
$\{1,2,4\}$     & 1 & 0.125 & 0.0125 & 0.00124\\
                 & 2 & 0.0827 & 0.00832 & 0.000829\\
$\{1,2,5,10\}$  & 1 & 0.0509 & 0.00513 & 0.000512\\
                 & 2 & 0.0632 & 0.00621 & 0.000636\\
\bottomrule
\end{tabular}
\end{table}

These numbers should not be read as a claim that percent-level fringe control
alone suffices in a real experiment. Ordinary environmental dephasing generally
depends on the level splitting as well, and fitting an additional scaling
exponent together with magnetic, electric, thermal, and readout nuisance terms
can substantially weaken identifiability. A credible test should tune one
transition within one apparatus, at fixed temperature and with common state
preparation and readout, so that the environmental contribution is as nearly
common-mode as possible. Zeeman or Stark tuning of a single transition is
preferable to comparing distinct transitions. The environmental scalings should
be characterized independently rather than assumed, since technical dephasing
from magnetic-field noise and from photon scattering carries its own known
dependence on splitting, and it is the residual after subtraction that
Eq. \eqref{eq:scaling_ratio} constrains. The interaction time should be varied
as well as the splitting, and residuals should be tested for curvature in
$\log\Gamma$ against $\log\lvert\Delta E\rvert$. Under the energy assignment, a
residual mass or spatial-separation dependence at fixed $\Delta E$ would show
that energy alone is not the relevant charge. Suitable platforms include
optical-lattice and single-ion clocks, Talbot--Lau matter-wave interferometry
with internal-state control \cite{fein2019}, and superconducting or spin qubits
with tunable splittings.

\subsection{Illustrative sensitivity benchmark}
\label{sec:results_sensitivity}

The following calculation is a hypothetical benchmark from stated
representative inputs, not a fit to published spectroscopic data. A real bound
requires the noise model of a specific experiment. Suppose a coherence with
$\Delta E=1.6\,\mathrm{eV}=2.56\times10^{-19}\,\mathrm J$ survives for
$T=10\,\mathrm s$ with any anomalous visibility loss bounded by one percent.
Then,
\begin{equation}
  \Lambda_\alpha
  \le \frac{-\log(0.99)}{\lvert\Delta E\rvert^{\alpha}T}
  =1.005\times10^{-3}\lvert\Delta E\rvert^{-\alpha}\,\mathrm s^{-1}.
  \label{eq:illustrative_bound}
\end{equation}
For $\alpha=2$ this gives
$\Lambda_2\le1.53\times10^{34}\,\mathrm{J}^{-2}\mathrm{s}^{-1}$, and for
$\alpha=1$ it gives
$\Lambda_1\le3.92\times10^{15}\,\mathrm{J}^{-1}\mathrm{s}^{-1}$.

Under the small-step identification of Section~\ref{sec:discussion_milburn},
the first value corresponds to an effective step
\begin{equation}
  \tau_{\mathrm M}^{\mathrm{eff}}
  =2\h^2\Lambda_2
  \le3.40\times10^{-34}\,\mathrm s,
  \label{eq:milburn_benchmark}
\end{equation}
approximately $6.3\times10^9$ Planck times. This is a unit conversion inside
the diffusive approximation, not an empirical bound on Planck-scale physics.
Note that the two coefficients are constrained on different scales and carry
different units, so a measurement at a single splitting can bound
$\Lambda_\alpha$ once $\alpha$ has been chosen but cannot estimate $\alpha$
itself. Only the scaling test of Eq. \eqref{eq:scaling_ratio}, with multiple
splittings, can do that, which is why Section~\ref{sec:results_precision} is the
operative part of the experimental program.

\section{Discussion}
\label{sec:discussion}

\subsection{Relation of the quadratic endpoint to established models}
\label{sec:discussion_milburn}

For the energy assignment $Q=H/\Estar$, the Markovian master equation at
$\alpha=2$ is
\begin{equation}
  \frac{\dd\rho}{\dd u}
  =-\frac{i}{\h}[H,\rho]
   -\frac{D}{\Estar^2}[H,[H,\rho]],
  \label{eq:milburn}
\end{equation}
since $[H,[H,\rho]]_{ab}=(E_a-E_b)^2\rho_{ab}$ in the energy eigenbasis. This
is a Gaussian energy-dephasing equation with the same double-commutator
structure that appears in several intrinsic-time and energy-decoherence models
\cite{adler2004,diosi2005}, but equality of master equations does not imply
equality of the microscopic constructions behind them.

Milburn's exact intrinsic-decoherence equation is the Poissonian random-unitary
jump model
\begin{equation}
  \frac{\dd\rho}{\dd u}
  =\frac1{\tau_{\mathrm M}}
   \left(e^{-iH\tau_{\mathrm M}/\h}\rho
   e^{iH\tau_{\mathrm M}/\h}-\rho\right).
  \label{eq:milburn_exact}
\end{equation}
For $\lVert H\rVert\tau_{\mathrm M}/\h\ll1$, expansion gives
\begin{equation}
  \frac{\dd\rho}{\dd u}
  =-\frac{i}{\h}[H,\rho]
   -\frac{\tau_{\mathrm M}}{2\h^2}[H,[H,\rho]]
   +O(\tau_{\mathrm M}^2),
  \label{eq:milburn_expansion}
\end{equation}
so Eq. \eqref{eq:milburn} matches Milburn's leading dissipative term when
$D/\Estar^2=\tau_{\mathrm M}/(2\h^2)$. At finite step size, however, the exact
coherence magnitude decays at rate
\begin{equation}
  \Gamma_{ab}^{\mathrm{Milburn}}
  =\tau_{\mathrm M}^{-1}
    \left[1-\cos\!\left(\frac{(E_a-E_b)\tau_{\mathrm M}}{\h}\right)\right],
  \label{eq:milburn_exact_rate}
\end{equation}
which is bounded and oscillatory and is not globally proportional to
$(E_a-E_b)^2$. The stable family should therefore be described as a scale-free
extension of the quadratic energy-dephasing generator, not as an exact
extension of Milburn's finite-step model. Equation \eqref{eq:milburn_exact_rate}
also shows why Theorem~\ref{thm:gaussian_nogo} must be stated narrowly:
Eq. \eqref{eq:milburn_exact} is itself a non-Gaussian phase model whose charge
dependence is not quadratic, so it is not the case that every
intrinsic-decoherence proposal in the literature is Gaussian.

Theorem~\ref{thm:gaussian_nogo} says only that when the accumulated phase is a
linear functional of a Gaussian field, the charge dependence is quadratic
irrespective of the covariance spectrum. Obtaining a strictly subquadratic
stable exponent then requires a non-Gaussian phase law, as in
Proposition~\ref{prop:poisson_stable}, or a random-clock construction such as
Eq. \eqref{eq:subordination_cf}. This is the single new physical hypothesis in
the framework, and its single new observable consequence is a decay rate
growing more slowly than quadratically in the charge or energy difference.

\subsection{Comparison with position-localization models}
\label{sec:discussion_csl}

The channel of Eq. \eqref{eq:lifted_channel} is pure dephasing in the
eigenbasis of $Q$. It is not a position-localization model and does not by
itself select an outcome. Continuous spontaneous localization instead modifies
trajectories in position space and typically predicts mass and separation
dependence as well as energy injection
\cite{ghirardi1990,bassi2013,carlesso2022}. Under $[H,Q]=0$ the present
dephasing channel conserves the energy expectation exactly, and at fixed
$\Delta E$ it predicts no dependence on particle mass or on the spatial extent
of the superposition unless the charge assignment introduces one. An experiment
sensitive to mass scaling therefore discriminates between the two
phenomenologies independently of the exponent.

Two qualifications are needed. First, as Eq. \eqref{eq:selection_energy} shows,
the exact energy conservation belongs to the dephasing channel, not to the
diffusive selection model of Eq. \eqref{eq:collapse_sse}, which heats whenever
$[H,N_a]\ne0$. A heating bound therefore constrains the selection sector and
not the dephasing sector. Second, a laboratory can host ordinary environmental
dephasing alongside any collapse-like channel, so a null bound on one mechanism
is not automatically a bound on the other.

\subsection{Model dependeces}
\label{sec:discussion_status}

The mathematically established core consists of the scale-free characterization
on the lifted line, the wrapped convolution semigroup for lattice-spaced
charges, the random-unitary channel, the universal complete-positivity boundary
with its three-charge obstruction, the quadratic charge dependence of Gaussian
integrated phases, the bounded-walk Born theorem, the Gaussian
It\^o/Stratonovich distinction together with the continuous-time statement of
section \ref{rem:marcus}, and the Wiener ensemble average of
Eq. \eqref{eq:collapse_lindblad}. The numerical work checks implementations and
illustrates estimator precision.
It is not used to prove any of these results.

Several statements are conditional. The inverse-power Poisson mechanism is
exact only for an ideal homogeneous point process with a singular kernel,
symmetric marks, and the convergence assumptions of
Proposition~\ref{prop:poisson_stable}.
Physical core sizes and finite volumes
regularize the model and generally leave only approximate stable scaling. The
subordination construction is exact as a stochastic process but becomes a
physical model only once a laboratory clock variable is identified. The energy
assignment $q_a=E_a/\Estar$ is phenomenological, since the channel theorem
itself allows any self-adjoint charge.

No selection dynamics is derived here. The bounded walk is a separate nonlinear
model whose Born probabilities follow from its martingale structure.
Equation \eqref{eq:collapse_sse} is a second, diffusive selection model. It
shows how one specific Gaussian unraveling yields a linear Lindblad ensemble,
but it neither proves that the stable dephasing noise selects outcomes nor
excludes carefully constructed non-Gaussian jump unravelings. There is
consequently no demonstrated single-noise theory unifying dephasing and
objective selection here, and the framework constrains a selection rule without
supplying one.

The precision calculation likewise has limited status. It quantifies
identifiability under an ideal two-parameter visibility model. A full
experimental proposal would require a platform-specific likelihood, a
calibration model, background channels, correlations, and a preregistered rule
for choosing the splitting range. Table~\ref{tab:precision} should be read as
demonstrating the value of a broad lever arm, not as a guaranteed significance
for an existing device.

\subsection{A bi-temporal geometry and its obstructions}
\label{sec:discussion_geometry}

It is natural to ask whether any physical structure supplies an unobserved
compact phase graded by a charge. A two-dimensional temporal plane provides a
formal compact angle, which in its
present form does not work well. Take $\R^{3,2}$ with
$x^M=(x,y,z,c\tau_1,c\tau_2)$ and $\eta_{MN}=\diag(-1,-1,-1,+1,+1)$. With
$\tau_1=t\cos\theta$ and $\tau_2=t\sin\theta$,
\begin{equation}
  \dd s^2=-\dd\mathbf x^2+c^2\dd t^2+c^2t^2\dd\theta^2.
  \label{eq:bitemporal_metric}
\end{equation}
The orbits at fixed $t$ are compact and timelike, so they are closed timelike
curves of proper duration $2\pi t$. This follows directly from
$g_{\theta\theta}>0$ and compactness and is a geometric obstruction, not a
coordinate artifact.

Appendix~\ref{app:geometry} derives the polar Dirac operator. In a rotating
spin frame, a spinor inherited from the unpunctured Cartesian plane is
antiperiodic around the circle. On the punctured plane there are two spin
structures, corresponding to periodic or antiperiodic rotating-frame spinors,
in the absence of an additional gauge holonomy. A gauge Wilson line can supply
a continuous twist, so the boundary phase is discrete only under that
proviso. The metric of Eq. \eqref{eq:bitemporal_metric} has no fixed
compactification radius, since its instantaneous circumference is $2\pi ct$ and
grows without bound. Replacing it by a constant $R$ is therefore a separate
frozen-radius diagnostic, not an adiabatic consequence of the geometry.

That diagnostic gives the formal dispersion relation
\begin{equation}
  E_n^2=\lvert\mathbf p\rvert^2c^2+m^2c^4
        -\left(\frac{\h c(n+\chi)}{R}\right)^2,
  \qquad \chi\in\left\{0,\half\right\},
  \label{eq:timelike_dispersion}
\end{equation}
the internal contribution entering negatively precisely because the compact
direction is timelike. A given mode is non-tachyonic only if
\begin{equation}
  R\ge\frac{\h\lvert n+\chi\rvert}{mc}.
  \label{eq:R_mode_bound}
\end{equation}
For $\chi=1/2$ the lowest mode requires $R\ge\h/(2mc)$, approximately
$1.9\times10^{-13}\,\mathrm m$ for an electron, which is excluded by many
orders of magnitude. More seriously, for every finite $R$ there are arbitrarily
large $\lvert n\rvert$ violating Eq. \eqref{eq:R_mode_bound}, so the entire
nonzero-mode tower is unstable on both branches unless an additional
ultraviolet truncation or gauge constraint is supplied. No relation between $R$
and the decoherence rate $D$ follows from this calculation.

The model has a second obstruction. Evolution on a $\tau_1=\mathrm{const}$
slice is not unitary in the naive positive $L^2(\dd^3x\,\dd\tau_2)$ norm,
because the second-time kinetic term is anti-Hermitian in the resulting
Hamiltonian. Ultrahyperbolic equations can have a well-posed codimension-one
initial-value problem on a nonlocally constrained set of data
\cite{craig2009,wang2022}, but no derivation identifies that constrained space
with a positive physical Hilbert space for the present spinor model.

The third obstruction is created by the geometry itself and is the most
consequential. Restricting to the $\theta$-independent sector, which exists
only for the periodic spin structure $\chi=0$, avoids propagation around the
closed timelike curve and removes the unstable tower. But it also removes all
nonzero internal charge differences and makes Eq. \eqref{eq:master_rate}
identically zero. Assigning instead $q_a=E_a/\Estar$ restores the phenomenology
only as an independent postulate on the real phase lift, and generic real
energy differences do not descend to the circle at all. The radial variable $t$
also fails to define a spacelike foliation of ordinary $\R^{3,1}$, so it does
not address the Lorentz-covariance objection to nonlocal reduction, which
concerns slicings of Minkowski space. Bars' two-time physics avoids analogous
ghosts through a local $Sp(2,\R)$ gauge symmetry that removes the extra
dimensions with no Kaluza--Klein remnant
\cite{bars2001,barsdeliduman2001,barskounnas1997}.
The construction considered
here uses a global $U(1)$ and has no corresponding gauge mechanism, and
conversely a model in which the extra time is entirely gauge-removable supplies
no residual stochastic phase and hence none of the phenomenology derived above.

For these reasons the bi-temporal construction is a source of geometric
intuition, not a consistent relativistic quantum field theory and not a
derivation of the dephasing law. Every theorem and every numerical result in
Sections~\ref{sec:methods_channel}--\ref{sec:results_sensitivity} is
independent of it, and rejecting the bi-temporal model loses none of the results. A viable
successor would need a gauge or constraint mechanism that removes the unstable
tower, a derived charge assignment, and a demonstrated relation between the
geometric radius and the stochastic rate.

\subsection{Falsification criteria}
\label{sec:discussion_falsification}

The framework admits several logically distinct tests.
First, an exact global rate law
$\Gamma(\Delta q)=D\lvert\Delta q\rvert^{\alpha}$ with $\alpha\notin(0,2]$
would falsify the symmetric scale-free L\'evy phase model, and for $\alpha>2$
it would also violate the universal complete-positivity condition of
Theorem~\ref{thm:cp_boundary}. A local effective slope above two over a finite
interval is a weaker observation and requires testing for crossover or
curvature before that conclusion can be drawn; in particular, a two-level
measurement at a single pair of splittings cannot on its own exhibit the
obstruction.

Second, reproducible curvature in $\log\Gamma$ against
$\log\lvert\Delta q\rvert$ would falsify exact scale covariance, though it
could remain compatible with a more general phase process. Anomalous population
transfer generated by the proposed channel would contradict unitality and
population preservation. Energy nonconservation in the dephasing sector would
falsify the combination of pure dephasing with $[H,Q]=0$, but not a model built
on a different charge; an anomalous rate tracking mass or superposition size
rather than energy splitting would refute the energy assignment independently
of $\alpha$.

Third, within the bounded reduction walk, a statistically significant departure
from Eq. \eqref{eq:born_walk} under verified symmetric, state-preserving,
nondegenerate proposals would refute the stated update or its implementation.
For Gaussian exact-flow dynamics, observation of genuine asymptotic absorption
without an additional term would contradict
Proposition~\ref{prop:convention}. Any demonstration of superluminal signaling
in a concrete realization would refute the construction through ensemble
linearity.

Finally, the geometric motivation can be eliminated by any requirement that a
physical nonzero tower coexist with positive, unitary, causal Hilbert-space
dynamics under the assumptions stated here. The present construction already
fails that requirement.

\section{Conclusion}
\label{sec:conclusion}

A symmetric L\'evy lift of an unobserved phase, coupled linearly to a charge,
generates the random-unitary dephasing law
$\Gamma_{ab}=D\lvert q_a-q_b\rvert^{\alpha}$. Requiring continuous scale
covariance of the L\'evy exponent fixes the stable family and restricts
$\alpha$ to $(0,2]$. The wrapped circle process is the appropriate compact
description only for lattice-spaced charge differences; arbitrary real charges
require the lift. Complete positivity over arbitrary finite charge spectra has
the same exact boundary, with three equally spaced charges providing a direct
analytic obstruction above two.

Gaussian integrated phases are necessarily quadratic in charge, though colored
Gaussian noise need not be Markovian. Ideal inverse-power Poisson fields and
subordinated Brownian motion give two mathematically explicit routes to
subquadratic stable exponents, subject to the physical qualifications stated in
the text.
The Poisson route is constrained rather than free, since the stable
window $0<\alpha<2$ coincides exactly with the convergence condition $p>d/2$.
The quadratic endpoint is the double-commutator energy-dephasing law and
matches the small-step expansion of Milburn's exact model, not its finite-step
dynamics.

For the separate bounded reduction walk, Born probabilities follow from a
bounded-martingale theorem together with a conditional-variance argument, and
the simulations confirm that they do not depend on the parent stable index.
Born-rule recovery therefore does not select the Cauchy case and places no
constraint on the exponent, which is consequently a quantity to be measured
rather than derived. In continuous time, neither the exact-flow Marcus reading
nor a naive It\^o jump equation reproduces this result. The first oscillates for
every $\alpha$, and the second is not state preserving. In the Gaussian
continuum limit the It\^o dynamics reduces with Born probabilities, whereas the
exact-flow/Stratonovich dynamics becomes a recurrent Brownian coordinate and
does not collapse without an imposed threshold. The Born rule in this model
class is thus convention dependent, and the two conventions give grossly
different and numerically distinguishable predictions.

The exponent is operationally accessible through ratios of anomalous rates at
different charge or energy splittings, a discriminator independent of the
unknown rate constant. The idealized calculation reported here shows that lever
arm matters as much as raw visibility precision, while also making clear that
realistic environmental nuisance channels must be modeled before any claim of
significance. The bi-temporal geometry does not supply the missing
microphysics. It has closed timelike curves, nonunitary naive slice evolution,
an unstable nonzero-mode tower, and no derived nonzero charge assignment. The
defensible claim is therefore the dephasing and martingale analysis, which
stands independently of any interpretation of the phase, together with the
measurement program that follows from it.

\newpage

\section*{Acknowledgments}
Partial support was provided by NIH and NSF funding.

\section*{Author contributions}
M.V.V. and I.D.D. developed the idea and wrote the manuscript.

\section*{Conflict of interest statement}
The authors declare no conflicts of interest.

\section*{Data and code availability}
This article does not report the collection or use of observational data. The
supplementary computational notebook contains the fixed random seeds and code
for all numerical checks in Section~\ref{sec:results}. The exponent-recovery
calculation uses the raw-visibility nonlinear fit specified in
Eq. \eqref{eq:precision_model}, and the archived notebook implements that
protocol. The benchmark in Section~\ref{sec:results_sensitivity} is arithmetic
from the stated hypothetical inputs and is not a fit to published data.

\newpage

\appendix
\setcounter{equation}{0}
\renewcommand{\theequation}{\thesection.\arabic{equation}}

APPENDIX

\section{Stable processes on the lift and on the circle}
\label{app:levy}

\subsection{Proof of Theorem~\ref{thm:wrapped_stable}}

For fixed $u>0$, the sequence $k\mapsto e^{-Du\lvert k\rvert^{\alpha}}$ is the
restriction to $\Z$ of the characteristic function of a symmetric
$\alpha$-stable random variable on $\R$, and is therefore positive definite on
$\Z$. Bochner's theorem for the compact abelian group $\Circ$ gives a unique
probability measure $\mu_u$ whose Fourier coefficients are
\begin{equation}
  \widehat\mu_u(k)=e^{-Du\lvert k\rvert^{\alpha}}.
  \label{eq:mu_fourier}
\end{equation}
The identity $\widehat\mu_{u+v}(k)=\widehat\mu_u(k)\widehat\mu_v(k)$ implies
$\mu_{u+v}=\mu_u*\mu_v$, so $(\mu_u)$ is a convolution semigroup. Since the
Fourier coefficients tend to one as $u\downarrow0$, the corresponding
convolution operators form a strongly continuous Markov semigroup.

For $u>0$ the Fourier coefficients are absolutely summable, so $\mu_u$ has the
continuous density of Eq. \eqref{eq:transition_density}. On the Fourier basis
$e_k(\theta)=e^{ik\theta}$,
\begin{equation}
  \lim_{u\downarrow0}\frac{T_ue_k-e_k}{u}
  =-D\lvert k\rvert^{\alpha}e_k,
  \label{eq:generator_modes}
\end{equation}
which identifies the generator as the closure of
$-D(-\Delta_\theta)^{\alpha/2}$ on its natural domain.

If a probability measure $\nu$ is invariant, then
$\widehat\nu(k)=e^{-Du\lvert k\rvert^{\alpha}}\widehat\nu(k)$ for every $u>0$,
so $\widehat\nu(k)=0$ for $k\ne0$ and $\widehat\nu(0)=1$, whence $\nu$ is Haar
measure. The nonzero Fourier modes decay exponentially in $u$, giving $L^2$ and
weak convergence. General background on stable processes and on L\'evy
processes on Lie groups is given in \cite{sato1999,liao2004}.

\subsection{Poisson-kernel identity}

At $\alpha=1$, with $r_u=e^{-Du}$,
\begin{align}
  p_u(\theta)
  &=\frac1{2\pi}\left(1+2\sum_{k\ge1}r_u^k\cos(k\theta)\right)\\
  &=\frac1{2\pi}\frac{1-r_u^2}{1-2r_u\cos\theta+r_u^2},
\end{align}
by the geometric-series identity. This is Eq. \eqref{eq:poisson_kernel}, the
wrapped Cauchy density of mean resultant length $r_u$ \cite{mardia2000}.

\subsection{Finite-dimensional distributions of the wrapped process}

Let $0=u_0<u_1<\cdots<u_N$ and $k_1,\ldots,k_N\in\Z$. Writing each
$\Theta_{u_j}$ as a sum of independent wrapped increments gives
\begin{equation}
  \E\!\left[\exp\!\left(i\sum_{j=1}^Nk_j\Theta_{u_j}\right)\right]
  =\prod_{j=1}^N
   \exp\!\left[-D\left\lvert\sum_{\ell=j}^Nk_\ell\right\rvert^{\alpha}
   (u_j-u_{j-1})\right].
  \label{eq:fdd}
\end{equation}
This family is consistent under deletion of time points and satisfies
Chapman--Kolmogorov, so it defines the wrapped L\'evy process. The restriction
of Eq. \eqref{eq:fdd} to integer arguments is precisely why arbitrary real
charge differences must be evaluated on the lift, and why the scale covariance
of Eq. \eqref{eq:scale_covariance} cannot be imposed on the circle.

\subsection{Drift and a static boundary twist}
\label{app:drift}

A deterministic drift $b$ changes the lifted characteristic function to
\begin{equation}
  \E[e^{i\xi(L_u-L_0)}]
  =\underbrace{e^{-Du\lvert\xi\rvert^{\alpha}}}_{\text{spreading}}
   \cdot\underbrace{e^{ib\xi u}}_{\text{rotation}} .
  \label{eq:drift_cf}
\end{equation}
The modulus of the coherence factor is unchanged, so drift produces a coherent
phase rotation rather than decoherence. A static boundary condition
$\Psi(\theta+2\pi)=e^{2\pi i\chi}\Psi(\theta)$ is a different object: it labels
an affine charge lattice $q_n=n+\chi$, and a common $\chi$ cancels from charge
differences. Neither affects $\Gamma_{ab}$, which depends only on charge
differences. Additional gauge holonomy can supply a continuous twist even
though the spin structures themselves are discrete.

\section{Complete positivity and the analytic obstruction above two}
\label{app:cp}

For a finite charge decomposition, Eq. \eqref{eq:dephasing_blocks} is the Schur
multiplier $\rho_{ab}\mapsto C^{(u)}_{ab}\rho_{ab}$. A Schur multiplier is
completely positive if and only if its multiplier matrix is positive
semidefinite, and its unit diagonal then makes the map trace preserving and
unital. The entrywise identity $C^{(u)}C^{(v)}=C^{(u+v)}$ gives the semigroup
property.

For $0<\alpha\le2$, the kernel $\lvert x-y\rvert^{\alpha}$ is conditionally
negative definite on $\R$, this being the L\'evy--Khintchine property of the
symmetric stable exponent \cite[Ch.~3]{sato1999}, and Schoenberg's theorem
\cite{schoenberg1938} makes its negative exponential positive definite for
every $u\ge0$. The argument uses nothing about the charges being integers or
even discrete, which is what licenses the energy assignment
$q_a=E_a/\Estar$; it is the appropriate complete-positivity argument for the
real-line lift.

For completeness the three-point obstruction is displayed directly. With
\begin{equation}
  q=(0,1,2),\qquad v=(1,-2,1),
  \label{eq:three_point_vectors}
\end{equation}
one has $\sum_av_a=0$ and
\begin{equation}
  v^{\mathsf T}Kv
  =\sum_{a,b}v_av_b\lvert q_a-q_b\rvert^{\alpha}
  =2(2^{\alpha}-4),
  \qquad K_{ab}=\lvert q_a-q_b\rvert^{\alpha}.
  \label{eq:three_point_form}
\end{equation}
For $\alpha>2$ this is positive, whereas a conditionally negative-definite
kernel requires it to be nonpositive. Since
$C^{(u)}=\mathbf1\mathbf1^{\mathsf T}-DuK+O(u^2)$ entrywise and
$v^{\mathsf T}\mathbf1\mathbf1^{\mathsf T}v=0$, the quadratic form of
$C^{(u)}$ is negative for sufficiently small $u$. This establishes failure with
no eigenvalue scan. Three charges are also the minimum: for a two-element
charge set the multiplier is
$\bigl(\begin{smallmatrix}1&c\\c&1\end{smallmatrix}\bigr)$ with
$c=e^{-Du\lvert\Delta q\rvert^{\alpha}}\in(0,1]$, which is positive
semidefinite for every exponent.

At $\alpha=2$ the Gaussian kernel is strictly positive definite for distinct
real charges and $u>0$. A numerically vanishing smallest eigenvalue for a large
and broad Gaussian Gram matrix reflects conditioning and finite precision, not
an exact zero eigenvalue.

\section{The bounded walk and Gaussian convention dependence}
\label{app:convention}

\subsection{Bloch-coordinate reduction}

Equation \eqref{eq:bloch_reduction} preserves $U_2^2+U_3^2$ because
\begin{equation}
  \frac{\dd}{\dd u}(U_2^2+U_3^2)
  =2U_2\dot U_2+2U_3\dot U_3=0 .
  \label{eq:bloch_invariant}
\end{equation}
On the unit circle, write $U_3=\cos\psi$ and choose the sign of $U_2=\sin\psi$
continuously on a branch; the deterministic equation is then separable away
from the fixed points. With
\begin{equation}
  z=\operatorname{arctanh}U_3=-\log\tan(\psi/2),
  \label{eq:logit_coordinate}
\end{equation}
if $w=\log\tan(\psi/2)$ obeys $w\mapsto w-\Delta X$, then the same exact flow
reads $z\mapsto z+\Delta X$.

\subsection{It\^o diffusion}

For Eq. \eqref{eq:ito_U}, $U_u$ is a bounded martingale. To identify its limit,
set $z=\operatorname{arctanh}U$; It\^o's formula gives
\begin{equation}
  \dd z=\sigma\dd W_u+\sigma^2\tanh z\,\dd u .
  \label{eq:ito_z}
\end{equation}
A scale function for this diffusion is $s(z)=\tanh z$, since
$s'(z)=\exp\bigl(-\int 2\sigma^2\tanh z\,\dd z/\sigma^2\bigr)
=\operatorname{sech}^2z$. On the finite interval $(-L,L)$ the probability of
exiting at $+L$ is therefore
\begin{equation}
  \Pr_z(\tau_{+L}<\tau_{-L})
  =\frac{s(z)+s(L)}{2s(L)}
  =\frac{U+\tanh L}{2\tanh L},
  \label{eq:ito_exit_finite}
\end{equation}
which tends to $(1+U)/2$ as $L\to\infty$. Equivalently, bounded-martingale
convergence gives an almost-sure limit, and a limit in the interior would leave
the quadratic-variation rate $\sigma^2(1-U^2)^2$ bounded away from zero,
contradicting convergence. Hence the limit is $\pm1$, with the stated Born
probability.

\subsection{Exact-flow and Stratonovich diffusion}

For Eq. \eqref{eq:strat_U} the ordinary chain rule applies, giving
\begin{equation}
  \dd z=\sigma\dd W_u,
  \qquad z=\operatorname{arctanh}U .
  \label{eq:strat_z}
\end{equation}
Brownian motion on $\R$ is recurrent and has no limit as $u\to\infty$, so
$U_u=\tanh(z_0+\sigma W_u)$ does not converge to either attractor. Converting
to It\^o form gives
\begin{equation}
  \dd U_u
  =\sigma(1-U_u^2)\dd W_u
   -\sigma^2U_u(1-U_u^2)\dd u,
  \label{eq:strat_ito_form}
\end{equation}
a drift toward the equator, so $U_u$ is a strict supermartingale for $U_u>0$
and a submartingale for $U_u<0$. On $(-L,L)$ the Brownian exit probability is
\begin{equation}
  \Pr_z(\tau_{+L}<\tau_{-L})
  =\frac{z+L}{2L}
  =\half+\frac{\operatorname{arctanh}U}{2L},
  \label{eq:strat_exit}
\end{equation}
which proves Eq. \eqref{eq:threshold_exit}. It is linear in
$\operatorname{arctanh}U_0$ rather than in $U_0$, and collapses to $\half$ as
the cutoff $L$ is increased.

For a discrete exact-flow walk with centered increments and accumulated
standard deviation $s_n$, the fixed-horizon probability $\Pr(U_n>0)$ is a
distribution-function calculation. It equals
$\Phi(\operatorname{arctanh}U_0/s_n)$ only when the accumulated increment is
Gaussian, or accurately approximated by one, and it should not be called an
absorption or gambler's-ruin probability. The two formulas agree only to first
order in $\operatorname{arctanh}U_0/s_n$.

\subsection{Continuous-time stable drivers}

Section \ref{rem:marcus} uses two standard facts. First, in the Marcus canonical
equation the noise acts through the exact deterministic flow of the vector
field $\sigma(1-U^2)\partial_U$, whose flow in the coordinate
$z=\operatorname{arctanh}U$ is translation; hence $U_u=\tanh(z_0+L_u)$
exactly. Second, a nondegenerate symmetric L\'evy process on $\R$ oscillates,
so that $\limsup_{u\to\infty}L_u=+\infty$ and
$\liminf_{u\to\infty}L_u=-\infty$ almost surely; symmetry excludes the drifting
alternatives of the classical trichotomy. Since $\tanh$ is a homeomorphism of
$\R$ onto $(-1,1)$, $U_u$ then oscillates between the two attractors without
converging, for every $\alpha\in(0,2]$.

For the naive It\^o jump equation $\dd U_u=\sigma(1-U_{u-}^2)\dd L_u$, a jump
of size $\Delta L$ produces $\Delta U=\sigma(1-U_{u-}^2)\Delta L$, which exits
$[-1,1]$ once $\lvert\Delta L\rvert$ is large enough; since a symmetric stable
L\'evy measure has unbounded support, this occurs with positive probability
from any interior state. For $\alpha\le1$ the driver is moreover not
integrable, so $\E[\Delta U\mid\mathcal F_{u-}]$ is undefined. The factor
$(1-U_{u-}^2)$ does not repair this, being $\mathcal F_{u-}$-measurable.

\section{Diffusive ensemble linearity and no-signaling assumptions}
\label{app:nosignaling}

\subsection{It\^o average of the trajectory equation}

Let $\overline N_a=\langle\Psi\vert N_a\vert\Psi\rangle$ and
$\Pi=\lvert\Psi\rangle\langle\Psi\rvert$. Applying It\^o's product rule to
Eq. \eqref{eq:collapse_sse} gives
\begin{equation}
  \dd\Pi
  =(\dd\lvert\Psi\rangle)\langle\Psi\rvert
   +\lvert\Psi\rangle(\dd\langle\Psi\rvert)
   +(\dd\lvert\Psi\rangle)(\dd\langle\Psi\rvert).
  \label{eq:ito_product}
\end{equation}
For each $a$, the drift and quadratic-variation terms combine as
\begin{align}
 &-\half\gamma_a^2
  \left[(N_a-\overline N_a)^2\Pi
       +\Pi(N_a-\overline N_a)^2\right]
  +\gamma_a^2(N_a-\overline N_a)\Pi(N_a-\overline N_a)\\
 &\hspace{25mm}
  =-\half\gamma_a^2[N_a,[N_a,\Pi]].
  \label{eq:ito_cancellation}
\end{align}
The scalar $\overline N_a$ cancels identically, because the double commutator
is invariant under $N_a\mapsto N_a-\overline N_a$. Terms linear in $\dd W^a$
have zero conditional expectation. Taking expectations therefore yields
Eq. \eqref{eq:collapse_lindblad}, in Lindblad form with self-adjoint
$L_a=\gamma_aN_a$. The It\^o counterterm also makes
$\dd\langle\Psi\vert\Psi\rangle=0$. This cancellation is what makes the
ensemble linear although each trajectory is not, and it is the step that
requires the Wiener quadratic variation.

A jump-driven normalized trajectory equation can also have a linear ensemble,
but its drift then contains the compensator associated with its own jump
measure. Such a model cannot be obtained by replacing $\dd W$ in
Eq. \eqref{eq:collapse_sse} while retaining the Wiener counterterm.

\subsection{Partial trace for local and additive dephasing}

Let $\hilb=\hilb_A\otimes\hilb_B$ and suppose
\begin{equation}
  H=H_A\otimes\id+\id\otimes H_B(\lambda),
  \label{eq:bipartite_H}
\end{equation}
where $\lambda$ is a local setting controlled by party $B$. A local
trace-preserving generator on $B$ disappears under $\Tr_B$. For an additive
self-adjoint dephasing operator
\begin{equation}
  L=A\otimes\id+\id\otimes B,
  \label{eq:additive_L}
\end{equation}
the double commutator expands into $A$--$A$, $B$--$B$, and cross terms. Partial
cyclicity gives
\begin{equation}
  \Tr_B([\id\otimes B,X])=0
  \label{eq:partial_cyclicity}
\end{equation}
for every trace-class $X$. Consequently the $B$--$B$ term vanishes after
$\Tr_B$, and so does the cross term, since
$\Tr_B([A\otimes\id,[\id\otimes B,\rho]])
=[A,\Tr_B([\id\otimes B,\rho])]=0$. Meanwhile
\begin{equation}
  \Tr_B([A\otimes\id,[A\otimes\id,\rho]])
  =[A,[A,\rho_A]],
  \qquad \rho_A=\Tr_B\rho .
  \label{eq:AA_term}
\end{equation}
Thus the reduced equation for $A$ contains only $H_A$, the $A$ parts of the
additive dephasing operators, and any explicitly local $A$ channels, and is
independent of $\lambda$. Iterating the Dyson series makes $\rho_A(u)$
independent of $\lambda$ for all $u$. The same argument applies with $A$ and
$B$ exchanged. It does not cover a genuinely nonlocal interaction Hamiltonian
or a product jump operator such as $A\otimes B$.

\section{Characteristic functions and subordination}
\label{app:influence}

\subsection{Cumulant form when moments exist}

If $X(u)$ has cumulants $\kappa_r(X(u))$ and the cumulant series converges in a
neighborhood containing the relevant charge difference, then
\begin{equation}
  \E[e^{i(q_a-q_b)X(u)}]
  =\exp\!\left[
    \sum_{r\ge1}\frac{i^r}{r!}(q_a-q_b)^r\kappa_r(X(u))
  \right].
  \label{eq:cumulant_expansion}
\end{equation}
For a centered Gaussian variable only $\kappa_2$ is nonzero, which recovers
Eq. \eqref{eq:gaussian_cf}; note, however, that the proof of
Theorem~\ref{thm:gaussian_nogo} does not rely on
Eq. \eqref{eq:cumulant_expansion} and requires no convergence hypothesis. A
non-Gaussian stable variable with $\alpha<2$ has no finite second moment, and
for $\alpha\le1$ no finite first absolute moment; its log-characteristic
function $-Du\lvert\xi\rvert^{\alpha}$ is non-analytic at the origin and
Eq. \eqref{eq:cumulant_expansion} does not converge term by term. The correct
statement in that case is the direct one,
$\E[e^{i\delta\phi_{ab}}]=e^{-Du\lvert q_a-q_b\rvert^{\alpha}}$.

\subsection{Proof of Proposition~\ref{prop:poisson_stable}}

For a finite annulus $r_0<\lvert x\rvert<R$, the Poisson exponential formula
and symmetry of the marks give
\begin{align}
 \log\E[e^{i\xi Y_{r_0,R}}]
 &=n\int_{r_0<\lvert x\rvert<R}
   \left(\E[e^{i\xi A\lvert x\rvert^{-p}}]-1\right)\dd x\\
 &=n\lvert\mathbb S^{d-1}\rvert
   \int_{r_0}^{R}
   \left(\E[\cos(\xi A r^{-p})]-1\right)r^{d-1}\dd r.
 \label{eq:poisson_log_cf}
\end{align}
As $r\to0$ the parenthetical factor is bounded by two and $r^{d-1}$ is
integrable, so no small-$r$ cutoff is required. As $r\to\infty$ the factor
behaves as $-\half\xi^2\E[A^2]r^{-2p}$ when the second mark moment is finite,
and the radial integrand as $r^{d-1-2p}$, which is integrable exactly when
$p>d/2$. For the nonnegative integrand $1-\cos(\cdot)$, Tonelli's theorem and
the assumption $\E\lvert A\rvert^{d/p}<\infty$ allow the radial and mark
integrals to be interchanged. The substitution $s=\lvert\xi A\rvert r^{-p}$
yields
\begin{align}
 &n\lvert\mathbb S^{d-1}\rvert
 \int_0^\infty
 \E[1-\cos(\xi A r^{-p})]r^{d-1}\dd r\\
 &\quad=\frac{n\lvert\mathbb S^{d-1}\rvert}{p}
   \E\lvert\xi A\rvert^{d/p}
   \int_0^\infty(1-\cos s)s^{-1-d/p}\dd s\\
 &\quad=C\lvert\xi\rvert^{d/p}.
\end{align}
The last integral converges at $s\to0$ because $1-\cos s\simeq s^2/2$, which
requires $d/p<2$, and at $s\to\infty$ because $d/p>0$; it is therefore finite
exactly for $0<d/p<2$. Sending $r_0\downarrow0$ and $R\uparrow\infty$ gives a
characteristic-function limit continuous at the origin, and L\'evy's continuity
theorem yields convergence in distribution together with
Eq. \eqref{eq:alpha_dp}. Ordinary moments of $Y$ can diverge even though this
characteristic-function limit exists. Conversely, imposing a nonzero hard core
changes the large-$\lvert Y\rvert$ behavior and destroys exact stability, and
taking $p\le d/2$ makes the centered annular sums square integrable so that the
central limit theorem returns a Gaussian limit.

\subsection{Subordination identity}

Conditioning on the clock $S_u$ gives
\begin{align}
  \E[e^{i\xi B_{S_u}}]
  &=\E\!\left[\E[e^{i\xi B_{S_u}}\mid S_u]\right]\\
  &=\E[e^{-\kappa\xi^2S_u}]\\
  &=e^{-cu(\kappa\xi^2)^{\beta}}
   =e^{-c\kappa^{\beta}u\lvert\xi\rvert^{2\beta}},
\end{align}
which proves Eq. \eqref{eq:subordination_cf}. This is the standard Bochner
subordination construction \cite{bochner1955}. The index of the resulting
symmetric stable phase is $\alpha=2\beta$, so the Cauchy case requires a
$1/2$-stable clock and the general case an $\alpha/2$-stable clock.

\section{Bi-temporal geometry: formal calculations and obstructions}
\label{app:geometry}

\subsection{Polar form of the Dirac operator}

Use gamma matrices satisfying $\{\gamma^M,\gamma^N\}=2\eta^{MN}\id$ for
$\eta=\diag(-1,-1,-1,+1,+1)$, and let
\begin{equation}
  \gamma^{\widehat t}(\theta)
  =\cos\theta\,\gamma^4+\sin\theta\,\gamma^5,
  \qquad
  \gamma^{\widehat\theta}(\theta)
  =-\sin\theta\,\gamma^4+\cos\theta\,\gamma^5.
  \label{eq:rotating_gammas}
\end{equation}
In polar coordinates on the temporal plane, the Cartesian Dirac equation is
\begin{equation}
 \left[
  i\h c\gamma^i\partial_i
  +i\h\gamma^{\widehat t}\partial_t
  +\frac{i\h}{t}\gamma^{\widehat\theta}\partial_\theta
  -mc^2
 \right]\Psi=0.
 \label{eq:polar_dirac_unrotated}
\end{equation}
Set $\Sigma=\gamma^4\gamma^5$, so that $\Sigma^2=-\id$, and
$S(\theta)=e^{-\theta\Sigma/2}=\cos(\theta/2)\id-\sin(\theta/2)\Sigma$. Then
$\gamma^{\widehat t}=S\gamma^4S^{-1}$ and
$\gamma^{\widehat\theta}=S\gamma^5S^{-1}$. With the rotating-frame spinor
$\Psi_{\mathrm{rot}}=S^{-1}\Psi$, and using
$S^{-1}\partial_\theta S=-\Sigma/2$ together with $\gamma^5\Sigma=-\gamma^4$,
Eq. \eqref{eq:polar_dirac_unrotated} becomes
\begin{equation}
 \left[
  i\h c\gamma^i\partial_i
  +i\h\gamma^4\left(\partial_t+\frac1{2t}\right)
  +\frac{i\h}{t}\gamma^5\partial_\theta
  -mc^2
 \right]\Psi_{\mathrm{rot}}=0.
 \label{eq:polar_dirac_rotated}
\end{equation}
The $1/(2t)$ term is the polar spin connection. The rescaling
$\wt\Psi=t^{1/2}\Psi_{\mathrm{rot}}$, which flattens the measure to
$\dd^3x\,\dd\theta$, removes it and gives
\begin{equation}
 \left[
  i\h c\gamma^i\partial_i
  +i\h\gamma^4\partial_t
  +\frac{i\h}{t}\gamma^5\partial_\theta
  -mc^2
 \right]\wt\Psi=0.
 \label{eq:kime_dirac}
\end{equation}
One may therefore display the spin connection explicitly as in
Eq. \eqref{eq:polar_dirac_rotated}, or absorb it by the density rescaling as in
Eq. \eqref{eq:kime_dirac}. Retaining both, or pairing the rotated spinor with
the $\theta$-dependent gamma matrices of Eq. \eqref{eq:rotating_gammas}, would
double count the connection.

\subsection{Spin structures and boundary conditions}

Because $S(2\pi)=e^{-\pi\Sigma}=-\id$, a Cartesian spinor that is single-valued
on the unpunctured temporal plane becomes antiperiodic in the rotating frame:
\begin{equation}
  \Psi_{\mathrm{rot}}(\theta+2\pi)=-\Psi_{\mathrm{rot}}(\theta).
  \label{eq:rotating_antiperiodic}
\end{equation}
On the punctured plane $\R^2\setminus\{0\}\simeq\Circ\times(0,\infty)$ and
$H^1(\Circ;\Z_2)=\Z_2$, so there are exactly two spin structures. In the
rotating frame they are represented by
\begin{equation}
  \wt\Psi(\theta+2\pi)=e^{2\pi i\chi}\wt\Psi(\theta),
  \qquad \chi\in\left\{0,\half\right\}.
  \label{eq:spin_boundary}
\end{equation}
The antiperiodic choice is the bounding structure inherited from the full
plane; the periodic choice is available only after the origin is excised.
Equation \eqref{eq:spin_boundary} classifies spin structures in the absence of
an additional flat gauge connection. A gauge Wilson line can add a continuous
holonomy, and that continuous parameter should not be conflated with the
discrete spin-structure choice. A common $\chi$ in any case cancels from
$\Gamma_{ab}$, which depends only on charge differences.

\subsection{Frozen-radius spectrum and tower instability}

The metric coefficient in Eq. \eqref{eq:bitemporal_metric} is time dependent.
To diagnose the sign of the angular contribution, freeze it at a chosen $t_0$
and write $R=ct_0$. For a mode $e^{i(n+\chi)\theta}$, squaring the
constant-coefficient diagnostic operator gives
Eq. \eqref{eq:timelike_dispersion}, and the effective four-dimensional mass
squared is
\begin{equation}
  M_n^2=m^2-\left(\frac{\h(n+\chi)}{Rc}\right)^2 .
  \label{eq:effective_mass}
\end{equation}
Each fixed mode therefore satisfies the non-tachyonic condition
Eq. \eqref{eq:R_mode_bound}, but no finite $R$ stabilizes the entire tower. For
$\chi=0$ the $n=0$ mode is unaffected while every sufficiently high nonzero
mode is unstable; for $\chi=1/2$ even the lowest-magnitude mode is shifted, and
the resulting bound $R\ge\h/(2mc)$ excludes that branch for Standard Model
fermions. Because the instability persists for arbitrarily large $\lvert n\rvert$
at any finite $R$, it cannot be evaded by choosing the spin structure. This
obstruction is independent of the decoherence parameter $D$, and no relation
between $R$ and $D$ follows from the calculation.

\subsection{Nonunitarity of a naive one-time slice}

Write the Cartesian equation as evolution in $\tau_1$,
\begin{equation}
  i\h\partial_{\tau_1}\Psi=\mathcal H\Psi,
  \label{eq:slice_evolution}
\end{equation}
with
\begin{equation}
  \mathcal H
  =-i\h c\gamma^4\gamma^i\partial_i
   -i\h c\gamma^4\gamma^5\partial_{(c\tau_2)}
   +mc^2\gamma^4 .
  \label{eq:slice_hamiltonian}
\end{equation}
Choose $\gamma^{4,5}$ Hermitian and $\gamma^{1,2,3}$ anti-Hermitian. Then
$(\gamma^4\gamma^i)^\dagger=\gamma^{i\dagger}\gamma^4=-\gamma^i\gamma^4
=\gamma^4\gamma^i$ is Hermitian, whereas
$(\gamma^4\gamma^5)^\dagger=\gamma^5\gamma^4=-\gamma^4\gamma^5$ is
anti-Hermitian. Since derivatives are anti-Hermitian under integration by
parts, the spatial kinetic and mass terms of Eq. \eqref{eq:slice_hamiltonian}
are Hermitian in the naive $L^2(\dd^3x\,\dd\tau_2)$ inner product while the
second-time kinetic term is anti-Hermitian. Hence
\begin{equation}
  \frac{\dd}{\dd\tau_1}\lVert\Psi\rVert^2
  =\frac{i}{\h}\langle\Psi\vert
   (\mathcal H^\dagger-\mathcal H)\vert\Psi\rangle
  \label{eq:norm_nonconservation}
\end{equation}
is nonzero in general. A constrained ultrahyperbolic initial-value theory may
use a different physical state space \cite{craig2009,wang2022}, but that
construction has not been carried out here.

\subsection{Zero-mode restriction and the charge problem}

On the periodic branch $\chi=0$, let $P_0$ denote projection onto the
$\theta$-independent angular sector. This is strictly stronger than projection
onto total angular charge zero: a charge-compensated pair of modes with charges
$+n$ and $-n$ is neutral but does not lie in $\Ran P_0$. The filter
$\cF_u$ of Eq. \eqref{eq:mode_filter} approaches $P_0$ in the normalized limit,
so the sharp projector arises as a limit of completely positive,
trace-non-increasing maps rather than as an external constraint; but no
interacting dynamics is shown to preserve that subspace, since a generic
interaction does not commute with the excitation number. On the antiperiodic
branch there is no $\theta$-independent spinor sector at all.

If only $\Ran P_0$ is retained, all angular charge differences vanish and the
compact-phase dephasing rate is identically zero, so the model becomes
empirically indistinguishable from ordinary quantum mechanics. The alternative
assignment $q_a=E_a/\Estar$ is not an angular tower charge and, for generic
real energy differences, does not descend to the circle; it is a
phenomenological charge for the real-line phase model. No equation in this work
derives that assignment or relates the geometric radius $R$ to the stochastic
rate $D$.

\newpage
\providecommand{\noopsort}[1]{}\providecommand{\singleletter}[1]{#1}%


\end{document}